\documentclass[12pt]{article}

\usepackage{amsmath,amssymb,amsthm,cite,mathtools,slashed}
\usepackage[margin=3cm]{geometry}
\usepackage{mathptmx,eucal} 
\usepackage{bm} 
\usepackage[greek,english]{babel}
\languageattribute{greek}{polutoniko}
\newcommand{\griego}[1]{{\selectlanguage{polutonikogreek}#1}} 

\title{The chirality theorem}

\author{
Jos\'e M. Gracia-Bond\'ia,$^1$
Jens Mund$^2$
and Joseph C.~V\'arilly$^3$
\\ [12pt]
{\footnotesize $^1$ Departamento de F\'isica Te\'orica,
Universidad de Zaragoza, Zaragoza 50009, Spain}\\[3pt]
{\footnotesize $^2$ Departamento de F\'isica,
Universidade Federal de Juiz de Fora, Juiz de Fora 36036--900, MG, 
Brasil}\\[3pt]
{\footnotesize $^3$ Escuela de Matem\'atica,
Universidad de Costa Rica, San Jos\'e 11501, Costa Rica}\\[3pt]
}

\date{\today}

\DeclareMathOperator{\supp}{supp}   
\DeclareMathOperator{\T}{T}         
\DeclareMathOperator{\tsum}{{\textstyle\sum}} 

\newcommand{\al}{\alpha}            
\newcommand{\bt}{\beta}             
\newcommand{\dl}{\delta}            
\newcommand{\eps}{\varepsilon}      
\newcommand{\ga}{\gamma}            
\newcommand{\ka}{\kappa}            
\newcommand{\La}{\Lambda}           
\newcommand{\la}{\lambda}           
\newcommand{\sg}{\sigma}            
\newcommand{\Th}{\Theta}            
\newcommand{\vf}{\varphi}           

\newcommand{\bM}{{\mathbb{M}}}      
\newcommand{\bR}{{\mathbb{R}}}      
\newcommand{\bS}{\mathbb{S}}        
\newcommand{\bW}{\mathbb{W}}        

\newcommand{\sD}{\mathcal{D}}       
\newcommand{\sO}{\mathcal{O}}       

\newcommand{\pt}{{\mathrm{p}}}      
\newcommand{\tree}{{\mathrm{tree}}} 

\newcommand{\del}{\partial}         
\newcommand{\delslash}{\slashed\del} 
\newcommand{\otto}{\leftrightarrow} 
\newcommand{\ovl}{\overline}        
\newcommand{\Psibar}{\overline{\Psi}} 
\newcommand{\w}{\wedge}             
\newcommand{\wh}{\widehat}          
\newcommand{\x}{\times}             
\renewcommand{\.}{\cdot}            
\newcommand{\7}{\dagger}            
\newcommand{\8}{\bullet}            

\newcommand{\half}{{\mathchoice{\thalf}{\thalf}{\shalf}{\shalf}}}
\newcommand{\ihalf}{\tfrac{i}{2}}   
\newcommand{\shalf}{{\scriptstyle\frac{1}{2}}} 
\newcommand{\thalf}{\tfrac{1}{2}}   

\bmdefine{\aaa}{a}          
\bmdefine{\bbb}{b}          
\bmdefine{\lll}{l}          
\bmdefine{\pp}{p}           

\newcommand{\llangle}{\langle\mkern-4mu\langle} 
\newcommand{\rrangle}{\rangle\mkern-4mu\rangle} 

\newcommand{\set}[1]{\{\,#1\,\}}     
\newcommand{\sump}{\sideset{}{'}\sum}   
\newcommand{\vev}[1]{\llangle#1\rrangle} 
\newcommand{\word}[1]{\quad\text{#1}\quad} 
\newcommand{\2}[1]{\underset{\star}{#1}} 

\def\wick:#1:{\,\mathopen:#1\mathclose:\,} 

\def\lLa^#1_#2{\Lambda^{#1}{}_{#2}}  
\def\ltLa_#1^#2{\Lambda_{#1}{}^{#2}} 

\newcommand{\twobytwo}[4]{\begin{pmatrix} 
         #1 & #2 \\ #3 & #4 \end{pmatrix}}

\newcommand{\tightldots}{\kern1pt.\kern0.5pt.\kern0.5pt.\kern1pt}

\theoremstyle{plain}
\newtheorem{thm}{Theorem}            
\newtheorem{prop}[thm]{Proposition}  
\newtheorem{lema}[thm]{Lemma}        
\newtheorem{corl}[thm]{Corollary}    
\newtheorem{schol}[thm]{Scholium}    

\numberwithin{equation}{section}

\makeatletter
\renewcommand{\section}{\@startsection{section}{1}{\z@}%
                       {-3.5ex \@plus -1ex \@minus -.2ex}%
                       {2.3ex \@plus.2ex}%
                       {\normalfont\large\bfseries}}
\renewcommand{\subsection}{\@startsection{subsection}{2}{\z@}%
                        {-3.25ex \@plus -1ex \@minus -.2ex}%
                        {1.5ex \@plus .2ex}%
                        {\normalfont\normalsize\bfseries}}
\makeatother

\begin{document}

\maketitle

\begin{abstract}
\smallskip
We show how chirality of the weak interactions stems from string
independence in the string-local formalism of quantum field theory.
\end{abstract}

\bigskip

\begin{quote}
\begin{flushright}

\griego{Swkr'aths --
<o no~un >'eqwn gewrg'os, ~<wn sperm'atwn k'hdoito ka`i >'egkarpa
bo'uloito gen'esjai, p'otera spoud~h| >`an j'erous e>is >Ad'wnidos
k'hpous >ar~wn\tightldots?}

-- Plato (\textit{Phaidros}, 276b)~\cite{Phaidros}

\end{flushright}
\end{quote}

\section{Introduction}
\label{sec:where-to}

Unanswered questions abound in electroweak theory \cite{Quigg09}. Only
time will tell which ones were prescient, and which born only from
theoretical prejudice~\cite{LynnS15}. A paramount trait of
flavourdynamics is the chiral character of the interactions in which
fermions and the massive vector bosons participate. A literature
search shows that most textbooks dispatch this trait in one word: it
is a \textit{fact}. There are a few exceptions. The book by Peskin and
Schroeder discusses at some length how left-handed and right-handed
components of fermions can come to see (representations of, if you
wish) different gauge groups~\cite[Chap.~19]{PeskinS95}. The
posthumous, reflective book by Bob Marshak \cite[Chaps.~1
and~6]{Marshak93}, discoverer (together with E.~C.~G. Sudarshan) of
the Vector-Axial theory, interestingly elevates the ``fact'' to a
principle, that of chirality invariance, or ``neutrino paradigm''.

Nevertheless, on the face of it, there is a mystery here, setting
flavourdynamics apart from chromodynamics. That cannot be solved by
invoking the Glashow--Weinberg--Salam (GWS) model, which introduces
chirality by hand from the outset.

The aim of this paper is to tackle this riddle through the theory of
string-local quantum fields (SLF). This conceptual framework was
introduced in~\cite{MundSY04,MundSY06}, improving on old proposals by
Mandelstam \cite{Mandelstam62} and Steinmann~\cite{Steinmann84}. It is
largely the brainchild of Schroer~\cite{Schroer16}.

At the considerable price of an extra variable, SL fields appear to
offer advantages over the ordinary sort. We summarily list them here.

\begin{itemize}

\item
The string-local fields evade the theorem that it is impossible
to construct on Hilbert space a \textit{vector} field for photons, and
more generally for corresponding representations associated to higher
fixed-helicity massless particles \cite[Sect.~5.9]{Weinberg95I}.

For this reason, the concept of gauge fades into the background.

\item
Other improved formal properties include a better ultraviolet
behaviour for spin and helicity~$>\half$; this turns out to be same
for all bosons as for scalar particles, and for fermions as for
spin-$\half$ particles.%
\footnote{Arguably, that is inherited from the amazingly good
behaviour of the field strengths themselves, beyond na\"ive power
counting, independently of spin, uncovered not long
ago~\cite{NikolovST14,Leto}.}
The upshot is that perturbative renormalization of SLF models should
take place without calling upon ghost fields, BRS invariance and the
like, since in principle one need not surrender positivity of the
energy and of the state spaces for the physical particles. It is fair
to say, however, that renormalization of theories with SL field theory
is still a work in progress.

\item
The reach of quantum field theory is enlarged, since the (boson and
fermion) Wigner unbounded-helicity particles \cite{Wigner39}, with
Casimirs $P^2 = 0$, $W^2 < 0$, that have no corresponding pointlike
fields~\cite{Yngvason70,IversonM71}, become admitted into the realm of
QFT through SL fields~\cite{MundSY04,MundSY06,Rehren17}.

\item
Furthermore, SLF proves its worth by shedding light on some
phenomenological conundrums of the current theory of fundamental
forces and particles. (Chief among them, after chirality, is the
observation that ``the SM accounts for, but does not explain,
electroweak symmetry breaking''~\cite{Peskin16}.)

\end{itemize}

We are going to show that the physical particle spectrum (charge and
mass structure) of the interaction carriers in the electroweak sector,
including the scalar particle, determines their relative coupling
strengths with the fermion sector entirely, and in particular
\textit{forces} the couplings of the massive bosons to fermions to be
parity-violating. 

In more detail, our input (particle and coupling types) is the 
experimental datum.
\begin{itemize}
\item
The \textit{particle types} are the electron, positron, neutrino and
antineutrino; the massive vector bosons $W_1$, $W_2$ and~$Z$, and the
photon; plus a scalar (Higgs) particle.%
\footnote{It will be enough here to consider just one generation of
leptons: bringing up the full structure of the fermion multiplets only
complicates the proof's notation in a way immaterial to the purpose.}

Their \textit{masses} obey $m_Z > m_W > 0$, and the photon is
massless. The electron and Higgs particle are massive; the masses
$m_e$, $m_\nu$ and $m_H$ are otherwise unconstrained, but are assumed
to be given.

The corresponding \textit{electric charges}: the
$Z$ and Higgs bosons, the neutrino and antineutrino $\nu,\bar\nu$ are
neutral; the electron $e$ and $W_-$ boson have charge~$-1$; the
positron $\bar e$ and $W_+$ have charge~$+1$. 

\item 
The \textit{couplings} are of two types. For the purely bosonic
couplings, see the beginning of Sect.~\ref{sec:bosonization}.

For the couplings between bosons and fermions we make the most general
Ansatz which respects electric charge conservation, Lorentz invariance
and renormalizability (scaling dimension $\leq 4$).
  
Apart from these general restrictions, our sole assumption is that in
photon-fermion couplings, the photon couples only to charged fermions,
so it does not couple to the neutrino or antineutrino; and even this
could be relaxed. All other coupling constants are left open.
\end{itemize}

Our powerful tool is the requirement that physical quantities like the
$\bS$-matrix must be independent of the string direction. This
principle is quite restrictive and, as we show here, in fact fixes all
coupling constants, bar the overall strength. In particular, it turns
out that:
\begin{itemize}
\item
the neutrino is completely chiral in that only left-handed%
\footnote{Or right-handed ones -- the theory of course cannot tell
which.}
neutrinos couple;
\item
the electron also couples in a parity-violating way;
\item
the Higgs particle couples only to scalar (and not to pseudoscalar)
Fermi currents.
\end{itemize}
This is our \textbf{chirality theorem}.

The proof, rigorous within perturbation theory, is achieved entirely
within the string-local scheme. It is simple, in that it requires only
consideration of tree graphs up to second order. Going
\textit{a~posteriori} from our framework to the GWS model for fermions
is both trivial and almost inconsequential; nevertheless, we indicate
how to do it in an appendix.

A valid argument for chirality, with the same outcome as ours, can be
made, and has indeed been made before, within the conventional
framework -- see \cite{DuetschS99,AsteSD99,Scharf16}; we owe these
works a lot. Apparently that proof was scarcely heeded, for reasons
not easy to understand. It is certainly couched in the language of
(the causal version of) gauge theory, keeping its ungainly retinue of
unphysical fields; and there is some circularity in it, since the
Kugo--Ojima asymptotic fields invoked \textit{ab initio} have to be
derived first. Our method provides a cleaner, more ``native'' form.
Still, theirs was a good case, and we are keen to employ new tools to
reclaim~it.

\medskip

The plan of the article is as follows. Section~\ref{sec:on-a-string}
is a pr\'ecis on free string-local fields.
Section~\ref{sec:perturbing} reviews the basics of perturbation theory
and Epstein--Glaser renormalization, as adapted to SLF, and introduces
the simple principle of physical \textbf{string-independence}
governing SLF couplings. The next two sections examine constraints
imposed on couplings with fermions by string independence already at
the first-order level. Section~\ref{sec:well-tempered} displays a
method, due to one of us, to construct time-ordered products involving
SLF for tree diagrams at second order.

Once that has been digested, the rest of the proof, performed in
Section~\ref{sec:at-laeva-malorum}, proceeds by a series of lemmas, of
interest in themselves, whose verifications reduce to fairly
straightforward calculations, entirely determining the couplings. In
particular, chirality of flavour\-dynamics emerges as an inescapable
consequence of string independence, given the mentioned physical
spectrum of intermediate vector bosons.
Section~\ref{sec:our-revels-now-are-ended} is the conclusion.

The supplementary sections deal with a few relevant side questions.
Appendices \ref{app:aliquid-obstat} and~\ref{app:cross-purposes}
furnish computational details. Appendix~\ref{app:BW} verifies locality
for the stringy fields. Appendix~\ref{app:y-ole} manufactures the GWS
model from the ascertained chiral coupling constants.

\section{String-local fields}
\label{sec:on-a-string}

To define the SLF, we start from free Faraday tensor fields on
Minkowski space~$\bM_4$. These can be built from Wigner's spin~$1$ or
helicity~$\pm 1$ unitary, irreducible representations of the
restricted Poincar\'e group~\cite{Wigner39}, by use of appropriate
creation operators $\al_r^\7(p)$ and polarization dreibein or zweibein
$e_r^\mu(p)$, under the form:
\begin{align}
F_a^{\mu\nu}(x) 
&:= \sum_r \int d\mu(p)\, \bigl[
e^{i(px)} \bigl( i p^\mu e_r^\nu(p) - i p^\nu e_r^\mu(p) \bigr) \,
\al_{r,a}^\7(p)
\notag \\
&\hspace{8em}
+ e^{-i(px)} \bigl( -i p^\mu e_r^\nu(p)^* +i p^\nu e_r^\mu(p)^* \bigr)
\, \al_{r,a}(p) \bigr],
\label{eq:Faraday-field} 
\end{align}
where $d\mu(p) := (2\pi)^{-3/2} \,d^3\pp/2E(\pp)$; we use the notation
$(ab) := g_{\la\ka} a^\la b^\ka = a^0b^0 - \aaa\cdot\bbb$ for
Minkowski inner products. Such fields are of the Lorentz
transformation type $(1,0) \oplus (0,1)$ -- see
\cite[Sect.~5.6]{Weinberg95I}. Consult also \cite{Stora05} in this
respect. Free string-local \textit{potential} fields are determined
from the~$F_a$:
\begin{equation}
A_a^\mu(x,l) := \int_0^\infty dt\, F_a^{\mu\la}(x + tl) \,l_\la \,,
\label{eq:A-vector-field} 
\end{equation}
with $l = (l^0,\lll)$ a null vector. By [half-]string we understand
the set of points $\{x + tl\}$, with $t \geq 0$. Each of the $A_a$
lives on the same Fock space as~$F_a$.

The main properties of the potential fields are as follows:
\begin{itemize}
\item
Transversality: $\bigl(l\,A_a(x,l)\bigr) = 0$; and
$\bigl(\del A_a(x,l)\bigr) = 0$ in the massless boson case.%
\footnote{Here and later, $(\del A) = \del_\mu A^\mu$ denotes a 
divergence.}
\item
Pointlike differential:
$\del^\mu A_a^\la(x,l) - \del^\la A_a^\mu(x,l) = F_a^{\mu\la}(x)$, or 
$dA_a=F_a$ for short.
\item
Covariance: let $U$ denote the second quantization of the mentioned
unitary representations of the restricted Poincar\'e group on the
one-particle states. Then
$$
U(a,\La) A_a^\mu(x,l) U^\7(a,\La) 
= A^\la(\La x + a, \La l)\, \ltLa_\la^\mu 
= \bigl( \La^{-1} \bigr)^\mu_{\;\;\la}\, A_a^\la(\La x + a, \La l).
$$
\item
Locality (causality): 
$[A_a^\mu(x,l), A_a^\la(x',l')] = 0$ when the strings $\{x + tl\}$ and
$\{x' + t'l'\}$ are causally disjoint.
\end{itemize}

The first three properties are nearly obvious. The last one is
subtler. It follows from (an easy variant of) the powerful argument
in~\cite{Mund09}, based on modular localization theory, spelled out in
Appendix~\ref{app:BW}.

Explicitly, in terms of~\eqref{eq:Faraday-field}, one finds that:
\begin{align}
A_a^\mu(x,l) 
&= \sum_r \int d\mu(p) \bigl[ e^{i(px)}\,u_r^\mu(p,l)\,\al_{r,a}^\7(p)
+ e^{-i(px)}\, u_r^{\mu}(p,l)^* \,\al_{r,a}(p) \bigr],
\notag \\
\word{with} u_r^\mu(p,l) 
&:= \int_0^\infty dt\, e^{it(pl)} i\bigl( p^\mu e_r^\la(p)
- p^\la e_r^\mu(p) \bigr) l_\la 
= e_r^\mu(p) - p^\mu \,\frac{(e_r(p)\,l)}{(pl)} \,.
\label{eq:A-field-explicit} 
\end{align}
Note that in the massless case, the denominator $(pl)$ may vanish;
nonetheless, $(e_r(p)\,l)/(pl)$ is locally integrable with respect to
the Lorentz-invariant measure $d\mu(p)$. In keeping with the
nomenclature of \cite{MundSY04,MundSY06}, the quantities
$u_r^\mu(p,l)$, $u_r^\mu(p,l)^*$, and similar ones for stringlike or
pointlike fields, are here called \textit{intertwiners}.

In this paper the set $\{F_a\}$ above includes one such field for each
of the physical particles, universally denoted $W^\pm$, $Z$, $\ga$.
For the massive ones, it does prove useful to consider the spinless
string-local \textit{escort}~fields:
\begin{equation}
\phi_b(x,l) := \sum_r \int d\mu(p)\, \biggl[
e^{i(px)} \frac{i(e_r(p)\,l)}{(pl)}\, \al_{r,b}^\7(p)
+ e^{-i(px)} \frac{-i(e_r(p)\,l)^*}{(pl)}\, \al_{r,b}(p) \biggr].
\label{eq:hire-an-escort} 
\end{equation}
We remark that
\begin{equation}
A^\mu_b(x,l) - \del^\mu \phi_b(x,l) =: A^{\pt,\mu}_b(x)
\label{eq:auxiliary-Proca} 
\end{equation}
defines pointlike \textit{Proca} fields, so that $dA^\pt_b = F_b$. All
these fields live on the same Fock spaces as the $F_b$ and have the
same mass. Moreover:
$$
\phi_b(x,l) = \int_0^\infty A^{\pt,\la}_b(x + sl) l_\la \,ds.
$$
Note the relations $(l\,\del\phi_b) = -(l A^\pt_b)$ and 
$$
\del_\mu A_b^\mu(x,l) + m_b^2 \phi_b(x,l) = 0.
$$
The last relation follows directly from \eqref{eq:A-field-explicit}
and~\eqref{eq:hire-an-escort}, since $(p\,e_r(p)) = 0$.

\medskip

Let now $d_l := \sum_\sg dl^\sg(\del/\del l^\sg)$ denote the
differential with respect to the string coordinate. We may introduce
the (form-valued in the string variable) field:
\begin{align}
d_l\phi_b(x,l) = w_b(x,l)
&:= \sum_r \int d\mu(p)\, \biggl[ e^{i(px)} \biggl(
\frac{ie_{r,\sg}(p)}{(pl)} - \frac{ip_\sg(e_r(p)l)}{(pl)^2} \biggr)\,
\al_{r,b}^\7(p)
\notag \\
&\qquad + e^{-i(px)} \biggl( \frac{ie_{r,\sg}(p)}{(pl)}
- \frac{ip_\sg(e_r(p)l)}{(pl)^2} \biggr)^* \, \al_{r,b}(p) \biggr]
\,dl^\sg;
\label{eq:good-form} 
\end{align}
and one obtains
\begin{align*}
\del_\mu w_b 
&= - \sum_r \int d\mu(p)\, \biggl[ e^{i(px)} \biggl(
\frac{p_\mu e_{r,\sg}(p)}{(pl)} - \frac{p_\mu p_\sg(e_r(p)l)}{(pl)^2}
\biggr)\, \al_{r,b}^\7(p)
\\
&\qquad + e^{-i(px)} \biggl( \frac{p_\mu e_{r,\sg}(p)}{(pl)}
- \frac{p_\mu p_\sg(e_r(p)l)}{(pl)^2} \biggr)^* \, \al_{r,b}(p)
\biggr] \,dl^\sg = d_l A^\mu_b\,;
\end{align*}
as well as $d_l w_b := d_l^2\phi_b = 0$.  In the case that $A^\mu_a$
describes a massless field, we just take the second equality in
\eqref{eq:good-form} as \textit{definition} of~$w_a$ 
and $d_l A_\ga^\mu = \del^\mu w_\ga$ still holds.%
\footnote{The form-valued $w_\ga$ suffers from expected infrared
problems. A promising way to deal with them in perturbation theory has
come to light recently~\cite{Duch17}.}

\medskip

We hasten now to exhibit a family of (Wightman) two-point functions
for our fields, of the general form
$$
\vev{\vf(x,l)\,\psi(x',l)} = \frac{1}{(2\pi)^3} \int d^4p\,
e^{-i(p(x-x'))} \dl_+(p^2 - m^2)\, M^{\vf\psi}(p,l)\,;
$$
where any of the two fields $\vf,\psi$, belong to the collection
$$
\set{ F_a^{\mu\nu}(x), \ A_a^\mu(x,l), \ \phi_b(x,l), \
\del^\mu\phi_b(x,l), \ w_a(x,l), \ \del^\mu w_a(x,l)} 
$$
with $a$ running over $(1,2,3,4)$ and $b$ over $(1,2,3)$. We shall
suppress the subindex notation $a,b$ in the rest of this section.
Here $\dl_+(p^2 - m^2) 
= \dl\bigl(p_0 - \sqrt{|\pp|^2 + m^2}\bigr)/2\sqrt{|\pp|^2 + m^2}$ and
$\vev{\text{\,---\,}}$ denotes a vacuum expectation value of the
included operator.

The respective $M^{\vf\psi}$ are computed from the definitions of
the fields. It is enough to note that:
$$
M^{\vf\psi}_{\al\bt} 
:= \sum_r u^{(\vf)}_{r,\al}(p,l)^* \, u^{(\psi)}_{r,\bt}(p,l),
$$
in terms of intertwiners $u^{(\vf)}$, $u^{(\psi)}$ already given. We
get, to begin with,
\begin{subequations}
\label{eq:two-pointers} 
\begin{equation}
M^{AA}_{\mu\nu} 
= - g_{\mu\nu} + \frac{p_\mu l_\nu + p_\nu l_\mu}{(pl)}\,.
\label{eq:two-pointers-AA} 
\end{equation}
The noteworthy and truly valuable fact here is that this is of
order~$0$ as $p^2 \to \infty$, while the two-point function of a
Proca field goes like~$p^2$. The formula is analogous to that which
comes out of lightcone gauge-fixing~\cite{Leibbrandt87}. However, the
meaning is quite different; in particular, our formalism is fully
covariant. On configuration space, therefore, $\vev{A(x,l)\,A(x',l)}$
essentially scales like $\la^{-2}$ under $x \mapsto \la x$, whereas
$\vev{A^\pt(x)\,A^\pt(x')}$ goes as~$\la^{-4}$.

Let us fill up a little table of vacuum expectation values of field
products, needed further down:
\begin{align}
M^{FF}_{\mu\nu,\rho\sg} 
&= -(p_\mu p_\rho \,g_{\nu\sg} - p_\nu p_\rho \,g_{\mu\sg}
- p_\mu p_\sg \,g_{\nu\rho} + p_\nu p_\sg \,g_{\mu\rho}).
\notag \\
M^{\del A,A}_{\mu\nu,\la} 
&= i\biggl( p_\mu g_{\nu\la} 
- p_\mu \frac{p_\nu l_\la + p_\la l_\nu}{(pl)}\,  \biggr),
\qquad 
M^{F\phi}_{\mu\nu} = \frac{p_\nu l_\mu - p_\mu l_\nu}{(pl)}\,,
\notag \\ 
M^{A\phi}_\mu &= - \frac{il_\mu}{(pl)}\,,  \qquad
M^{A,\del\phi}_{\mu\nu} = \frac{p_\nu l_\mu}{(pl)}\,, \qquad
M^{\del A,\phi}_{\mu\nu} = - \frac{p_\mu l_\nu}{(pl)}\,, 
\notag \\
M^{\phi\phi} &= \frac{1}{m^2}\,,  \qquad\quad
M^{\del\phi,\phi}_\mu = - \frac{ip_\mu}{m^2}\,,
\label{eq:two-pointers-more} 
\end{align}
as well as
\begin{align}
M^{Aw}_\mu &= \frac{i}{(pl)}\, M^{AA}_{\mu\sg} \,dl^\sg
= i \biggl( \frac{-g_{\mu\sg}}{(pl)}
+ \frac{p_\sg l_\mu}{(pl)^2} \biggr) \,dl^\sg \,,
\qquad M^{w\phi} = 0,
\notag \\
M^{ww} &= \frac{1}{(pl)^2}\, M^{AA}_{\sg\tau} \,dl^\sg \w dl^\tau
= - \frac{g_{\sg\tau}}{(pl)^2} \,dl^\sg \w dl^\tau,
\notag \\
M^{Fw}_{\mu\nu} &= d_l\,M^{F\phi}_{\mu\nu} 
= i \biggl( \frac{p_\nu g_{\mu\sg} - p_\mu g_{\nu\sg}}{(pl)}
+ \frac{p_\mu l_\nu - p_\nu l_\mu}{(pl)^2} p_\sg \biggr) \,dl^\sg \,,
\label{eq:two-pointers-w} 
\end{align}
\end{subequations}
using the relation $l_\sg \,dl^\sg = 0$. It is clear that massless
bosons do not bear escort quantum fields.%
\footnote{Spacelike strings have been more often employed in the
literature on SLF. It is nevertheless better here to deal with
lightlike strings, since then in general the intertwiners are
functions, not just distributions; so we need not smear them. Our
arguments work either way~\cite{Figueiredo17}.}

\medskip

The construction of SLF for spin $2$ or helicity $\pm 2$ proceeds in
the same way, from the equivalent object to the Faraday tensor~$F$,
the linearized Riemann tensor $R$ for spin or helicity~$2$, towards
the string-local replacement for the pointlike (symmetric rank~$2$
tensor) ``potential''.  Note that physical scalar fields are not
stringy.%
\footnote{Nor are free Dirac fields; SLF for half-integer spin greater
than $\half$ or integer spin greater than~$2$ are discussed
elsewhere~\cite{PlaschkeY12,MundDO17,MundRS17b}.}

\section{Perturbation theory for SLF: the role of string independence}
\label{sec:perturbing}

New theories demand care with the mathematics. We intend to borrow
from the St\"uckelberg--Bogoliubov--Epstein--Glaser (SBEG)
``renormalization without regularization'' formalism for perturbation
theory, both most rigorous and flexible
\cite{BogoliubovS80,EpsteinG73}. Since renormalization theory for SLF
is in its infancy, it still works partly as a heuristic guide. We
only outline what we need here from it.

The method involves the construction of a scattering operator
$\bS[g;l]$ functionally dependent on a (multiplet of) smooth external
fields $g(x)$, which mathematically are test functions. The procedure
is natural in view of locality; the functional scattering operator
acts on the Fock spaces corresponding to local free fields, of the
pointlike or stringlike variety, for a prescribed set of free
particles. It is submitted to the following conditions.

\begin{itemize}
\item
Covariance: 
$U(a,\La) \bS[g;l] U^\7(a,\La) = \bS[(a,\La)g; \La l]$, with
$(a,\La)g(x) = g(\La^{-1}(x - a))$.

\item
Unitarity: $\bS^{-1}[g;l] = \bS^\7[g;l]$.

\item
Causality. Let $V^+$, $V^-$ denote the future and past solid light
cones. Then
\begin{equation}
\bS[g_1 + g_2;l] = \bS[g_1;l]\, \bS[g_2;l]
\label{eq:madre-del-cordero} 
\end{equation}
when 
$(\supp g_2 + \bR^+l) \cap (\supp g_1 + \bR^+l + V^+) = \emptyset$,
or equivalently 
$(\supp g_1 + \bR^+l) \cap (\supp g_2 + V^- + \bR^+l) = \emptyset$.
\end{itemize}

In practice one looks for $\bS[g;l]$ as a power series in~$g$, of the
form
\begin{equation}
\bS[g;l] = 1 + \sum_{k=1}^\infty \frac{i^k}{k!} \int_{\bM_4^k}
S_k(x_1,\dots,x_k,l) g(x_1) \cdots g(x_k) \,dx_1 \cdots dx_k.
\label{eq:S-matrix} 
\end{equation}

Only the first-order term $S_1$ is postulated. This will be a Wick
polynomial in the free fields.%
\footnote{In many models it looks like an interaction Lagrangian. It
should, however, be kept in mind that the building blocks in the
procedure are quantum fields; ditto, our starting point is Wigner's
theory of quantum Poincar\'e modules~\cite{Wigner39} and corresponding
field-strength representations of the Lorentz group, rather than a
classical Lagrangian that one attempts to ``quantize''.}

We come back in a moment to the structure of $S_1$ in the present
context. In consonance with \eqref{eq:madre-del-cordero}, the
$S_k(x_1,\dots,x_k,l)$ for $k \geq 2$ are \textit{time-ordered
products}, which need to be constructed. By locality, the causal
factorization
\begin{equation}
S_2(x,x',l) = \T[S_1(x,l) S_1(x',l)]
:= S_1(x,l) S_1(x',l) \text{ \ or \ } S_1(x',l) S_1(x,l),
\label{eq:T2-factor} 
\end{equation}
according as $\{x + tl\}$ is later or earlier than $\{x' + tl\}$,
fixes $S_2$ on a large region of $\bM_4^2 \x S^2$. Indeed, assuming
$l^0 > 0$, a string $\{x+tl\}$ lies to the future of another string
$\{x'+t'l\}$ if and only if $((x - x')\,l) \geq 0$ and the
intersection of the strings is empty. That is, $x$ lies to the future
of, or~on, the hyperplane $x' + l^\perp$, but not on the full line
$x'+ \bR l$ \cite{Figueiredo17}. Consequently, the strings cannot be
ordered if and only if $x$ lies on the string $\{x'+t'l\}$ or vice
versa; i.e., if and only if $x - x'$ is lightlike and parallel to~$l$.
This exceptional set:
\begin{equation}
\sD := \set{(x,x',l) : (x - x')^2 = 0, \ ((x - x')\,l) = 0}
\label{eq:bad-to-the-bone} 
\end{equation}
is of measure zero in $\bM_4^2 \x S^2$. The extension of such products
to the whole of $\bM_4^2 \x S^2$, mainly by upholding string
independence, is the SBEG renormalization problem in a nutshell.

Existence of the adiabatic limit is the property that the $S_k$ be
integrable distributions, in the sense of Schwartz~\cite{Schwartz66}.
In that limit, as $g$ goes to a constant, the covariant $\bS[g;l]$
is expected to approach the invariant physical scattering
matrix~$\bS$, so that $U(a,\La) \bS\, U^\7(a,\La) = \bS$, all
dependence on the string disappearing.

\medskip

A lesson of gauge field theory is that couplings of quantum fields
should fall out from \textit{a simple underlying principle}. The
natural and essential hypothesis of interacting SLF theory is simple
enough: physical observables and quantities closely related to them,
particularly the $\bS$-matrix, cannot depend on the string
coordinates. This is the \textbf{string-independence} principle:
colloquially, the string ``ought not to be seen''. Let $S_1$ denote a
first-order vertex coupling in general. For the physics of the model
described by $S_1$ to be string-independent, one must require that a
vector field $Q_1^\mu(x,l)$ exist such that
\begin{equation}
d_l S_1 = (\del Q_1) \equiv \del_\mu Q_1^\mu \,,
\label{eq:yet-unseen} 
\end{equation}
so that, regarding the $\bS$-matrix as the adiabatic limit of
Bogoliubov's functional $\bS$-matrix, on applying integration by
parts, the contribution from the divergence vanishes. Moreover,
(perturbative) string independence should hold at every order in the
couplings, surviving renormalization.

Already the condition that $d_l S_1$ be a divergence severely
restricts the interaction vertices in~$S_1$; we proceed to throw light
on the fermion sector by using it in the next section. Further along,
all the time-ordered products $S_k$ in the functional $\bS$-matrix
ought to be determined from string independence.

\section{On the string-local boson sector}
\label{sec:bosonization}

It turns out that the string independence principle holds great power
both as a heuristic device and a justification tool, dictating
\textit{symmetry} (of the Abelian and non-Abelian kind) from
interaction%
\footnote{Thus reversing Yang's \textit{dictum}, restated in the
famous terminological discussion on gauge interactions between Dirac,
Ferrara, Kleinert, Martin, Wigner, Yang himself and Zichichi
\cite{Zichichi84}.}
down to almost every nut and bolt. A complete account of electroweak
theory would start by showing that, when the string independence
principle is applied to the physically relevant set of boson SLF, with
their known masses and charges, replacing the standard pointlike
fields, plus one \textit{physical} Higgs particle $\phi_4(x)$,%
\footnote{Following Okun~\cite{Okun91}, and for obvious grammatical
reasons, henceforth we refer to a (physical) Higgs boson as a higgs,
with a lowercase~h. Note also that, in the presence of a
massless~$A_4$, the notation~$\phi_4$ is not meant to purport the
higgs as a rogue escort!}
one recovers precisely the phenomenological couplings of
flavourdynamics in the Standard Model (SM), with massive bosons
mediating the weak interactions, and the $U(2)$ structure constants,
as, for instance, in \cite{Scheck12} or~\cite[Ch.~1]{Nagashima13}.
(One cannot quite say that we recover \textit{the} Standard Model
picture after spontaneous symmetry breaking has allegedly taken place,
since our boson fields are different, and our rule set cares little
for Lagrangians. But the coincidence of the couplings ought to be
evident -- see the discussion at the end of
Section~\ref{sec:at-laeva-malorum}.)

Such a derivation, spelled out in a future paper~\cite{WienNotes},
requires one to examine time-ordered products corresponding to graphs
involving boson particles up to third order in the couplings. For want
of space, here we can just display its flavour, and foremost the
results we need, to build up our derivation for chirality of weak
interactions.

\begin{itemize}

\item
Apart from the higgs particle sector, a string-local theory of
interacting bosons at first order in the coupling constant~$g$ must be
of the form:
\begin{align}
S^B_1(x,l) &= g \sum_{a,b,c} f_{abc} F_a(x) A_b(x,l) A_c(x,l)
\label{eq:bosons-mate} 
\\
& + g \sump_{a,b,c} f_{abc} (m^2_a - m_b^2 - m_c^2) 
\bigl( A_a(x,l) A_b(x,l) \phi_c(x,l)
- A_a(x,l) \del\phi_b(x,l) \phi_c(x,l) \bigr),
\notag 
\end{align}
where we omit the notation $\wick:\text{\,---\,}:$ for Wick products,
and the restricted sum $\sum'$ runs over massive fields only. Here the
$f_{abc}$ denote the (completely skewsymmetric) structure constants of
the (reductive) symmetry group of the model; the mass of the vector
boson $A_a$ is denoted $m_a$, and complete contraction of Lorentz
indices is understood. Notice that the escort fields hold a somewhat
analogous place to St\"uckelberg fields.

\item
Now it is straightforward to check that the $1$-form $d_l S^B_1$,
measuring the dependence on the string variable of the vertices
in~\eqref{eq:bosons-mate}, is a divergence: 
$d_l S^B_1(x,l) = \bigl(\del Q_1^B \bigr)(x,l)$, where $Q_1^B$ is
given by:
\begin{equation}
2g \sum_{a,b,c} f_{abc} (F_a A_c)_\mu w_b 
+ g \sump_{a,b,c} f_{abc} (m_a^2 + m_c^2 - m_b^2)
(A_{a\mu} - \del_\mu\phi_a) \phi_c w_b.
\label{eq:bosons-shadow} 
\end{equation}
We shall need $Q_1^B$ to prove chirality of the couplings to the
fermion sector. 

\item
At once we adapt our notation to the one used in the SM. This model
has three masses $m_1 = m_2 < m_3$ different from zero and one
$m_4 = 0$. Defining the Weinberg angle%
\footnote{This makes sense in the renormalized theory
\cite[Sect.~29.1]{Schwartz14}.}
by $m_1/m_3 =: \cos\Th$, we employ the basis in which
$$
f_{123} = \half \cos\Th, \quad  
f_{124} = \half \sin\Th, \quad  f_{134} = f_{234} = 0,
$$
all other $f_{abc}$ following from complete skewsymmetry. They are
seen to be the structure constants of (the Lie algebra of) the $U(2)$
determined by the \textit{physical} particle fields. We shall use the
standard notations
$$
W_\pm \equiv \frac{1}{\sqrt{2}}(W_1 \mp iW_2)
:= \frac{1}{\sqrt{2}}(A_1 \mp iA_2), \quad
Z := A_3,  \quad  A := A_4  
$$
and similarly for $\phi_\pm$, $w_\pm$, $\phi_Z$ and~$w_Z$; with masses
$m_W = m_1$, $m_Z = m_3$ and $m_\ga = m_4 = 0$.

\item
With this in hand, we focus on \eqref{eq:bosons-shadow}, keeping in
mind that, although an escort field does not exist for the photon, the
field $w_4$ exists at the same title as $w_1$, $w_2$ and $w_Z$. The
first summand in~\eqref{eq:bosons-shadow} yields:
\begin{align}
& 2g \tsum f_{abc}(\del_\mu A_{a\la} - \del_\la A_{a\mu}) A_c^\la w_b
\notag \\
&\quad = ig \sin\Th \bigl[ (\del_\mu A_\la - \del_\la A_\mu)
(w_- W_+^\la - w_+ W_-^\la)
+ (\del_\mu W_{-\la} - \del_\la W_{-\mu})(w_+ A^\la - w_4 W_+^\la)
\notag \\
&\hspace*{6em} 
+ (\del_\mu W_{+\la} - \del_\la W_{+\mu})(w_4 W_-^\la - w_- A^\la)
\bigr] 
\notag \\
&\quad + ig \cos\Th \bigl[(\del_\mu Z_\la - \del_\la Z_\mu)
(w_- W_+^\la - w_+ W_-^\la) 
+ (\del_\mu W_{-\la} - \del_\la W_{-\mu})(w_+ Z^\la - w_Z W_+^\la) 
\notag \\
&\hspace*{6em}
+ (\del_\mu W_{+\la} - \del_\la W_{+\mu})(w_Z W_-^\la - w_- Z^\la)
\bigr].
\label{eq:capeat-qui-potest} 
\end{align}

\item
Our $Q^B_{1\mu}$ above is not complete, since bosonic couplings
involving the higgs sector have not been included. They are also 
derived from the string-independence principle.%
\footnote{There again, SLF theory goes one better: the ``negative
squared mass'' in the higgs' self-potential, not accounted for in the
SM~\cite{Peskin16}, is derived from string independence. We refer to
the forthcoming \cite{MundS17} in this respect.}
Of those, for our purposes in this paper we need only:
\begin{equation}
\frac{g}{2\cos\Th} m_Z (\phi_4 (\del_\mu\phi_Z - Z_\mu)
- \del_\mu\phi_4\, \phi_Z) w_Z;
\label{eq:higgs-shadow-terms} 
\end{equation}
actually these play a pivotal role in our problem. Clearly, terms of
this type are suggested by the last group of summands in
\eqref{eq:bosons-shadow}.

\item
By the way, the expected $g^2 AAAA$ terms and thus the indications of
the classical geometrical gauge approach are recovered in our
formalism from string independence at the level of~$S_2$.

\end{itemize}

\section{The first-order constraints}
\label{sec:so-far-so-good}

Our framework for electroweak theory is outlined next. This both
exemplifies the principle and contributes to the core of this paper.

\begin{itemize}

\item
The couplings between interaction carriers and matter currents in a
theory with massive or massless vector bosons $A_{a\mu}$ must be of
the form
\begin{gather}
g\bigl (b^a A_{a\mu} J_V^\mu + \tilde b^a A_{a\mu} J_A^\mu 
 + c^a \phi_a S + \tilde c^a \phi_a S_5 \bigr) ; 
\label{eq:tapa-del-perol} 
\\
\shortintertext{where}
J_V^\mu = \ovl\psi \ga^\mu \psi, \quad
J_A^\mu = \ovl\psi \ga^\mu \ga^5 \psi, \quad 
S = \ovl\psi \psi, \quad 
S_5 = \ovl\psi \ga^5 \psi,
\notag
\end{gather}
with electric charge conserved in the interaction vertices. Our key
assumption point is that these $A_a^\mu$ and~$\phi_a$ above are now
given as string-local quantum fields, thus satisfying
renormalizability by power counting. There exist \emph{no other scalar
couplings which comply with renormalizability}. To wit, Lorentz
invariance requires that all cubic terms be of the above form, and
renormalizability forbids quartic terms.%
\footnote{Since the two Fermi fields required by Lorentz invariance
already have scaling dimension~$3$, any two further fields would
give~$5$, exceeding the power-counting limit.}

\item
The $\psi$ in \eqref{eq:tapa-del-perol} are ordinary fermion fields --
we should not assume chiral fermions \textit{ab~initio}, and we
do~not.

\item
The coefficients $b^a$, $\tilde b^a$, $c^a$, $\tilde c^a$
in~\eqref{eq:tapa-del-perol} are to be determined from string
independence.

\end{itemize}

The proof of chirality in the couplings of electroweak bosons to the
fermion sector of the SM from string independence develops in two
stages. In the first stage, we need not invoke the $Q_1$-vector of the
boson sector. For these couplings, we make the most general Ansatz, as
explained after~\eqref{eq:tapa-del-perol}, again omitting the notation
$\wick:\text{\,---\,}:$ for the Wick products:
\begin{align}
S_1^F(x,l) 
&:= g\bigl( b_1 W_{-\mu} \bar e \ga^\mu \nu 
+ \tilde b_1 W_{-\mu} \bar e \ga^\mu \ga^5 \nu
+ b_2 W_{+\mu} \bar\nu \ga^\mu e
+ \tilde b_2 W_{+\mu} \bar\nu \ga^\mu \ga^5 e
\notag \\
&\qquad + b_3 Z_\mu \bar e \ga^\mu e
+ \tilde b_3 Z_\mu \bar e \ga^\mu \ga^5 e
+ b_4 Z_\mu \bar\nu \ga^\mu \nu
+ \tilde b_4 Z_\mu \bar\nu \ga^\mu \ga^5 \nu
\notag \\
&\qquad + b_5 A_\mu \bar e \ga^\mu e 
+ \tilde b_5 A_\mu \bar e \ga^\mu \ga^5 e 
+ b_6 A_\mu \bar\nu \ga^\mu \nu 
+ \tilde b_6 A_\mu \bar\nu \ga^\mu \ga^5 \nu
\notag \\
&\qquad + c_1 \phi_- \bar e \nu + \tilde c_1 \phi_- \bar e \ga^5 \nu
+ c_2 \phi_+ \bar\nu e + \tilde c_2 \phi_+ \bar\nu \ga^5 e
\notag \\
&\qquad + c_3 \phi_Z \bar e e + \tilde c_3 \phi_Z \bar e \ga^5 e
+ c_4 \phi_Z \bar\nu \nu + \tilde c_4 \phi_Z \bar\nu \ga^5 \nu
\notag \\
&\qquad + c_0 \phi_4 \bar e e + \tilde c_0 \phi_4 \bar e \ga^5 e
+ c_5 \phi_4 \bar\nu \nu + \tilde c_5 \phi_4 \bar\nu \ga^5 \nu \bigr).
\label{eq:fermis-wish-list} 
\end{align}

All the boson fields here are string-local, except for the pointlike
higgs field $\phi_4$. Here $e$~stands for an electron, or muon, or
$\tau$-lepton pointlike field or for (a suitable combination of) quark
fields $d,s,b$; and $\nu$ for the neutrinos or for the quarks
$u,c,t$.%
\footnote{As already indicated, we consider just one generation of
leptons.}
Charge is conserved in each term. Unitarity of the $\bS$-matrix, in
the light of~\eqref{eq:S-matrix}, dictates that $S_1$ be Hermitian.
Thus, for instance, $b_2 = b_1^*$ and $\tilde b_2 = \tilde b_1^*$; and
we may choose phases so that both $b_1$ and~$\tilde b_1$ are real.
Moreover, $b_3, b_4, b_5, b_6$ and 
$\tilde b_3, \tilde b_4, \tilde b_5, \tilde b_6$ are all real; 
$c_2 = c_1^*$ and $\tilde c_2 = - \tilde c_1^{\,*}$; 
$c_3, c_4, c_0, c_5$ are real, whereas 
$\tilde c_3, \tilde c_4, \tilde c_0, \tilde c_5$ are imaginary. We may
assume that the photon should not couple to neutrinos, which are
uncharged, and drop the corresponding terms, with coefficients $b_6$,
$\tilde b_6$, right away.%
\footnote{Were we not to do so, electric charge would appear as the
\textit{difference} between the couplings of the photon to the
electron and the neutrino.}

As indicated, the $\psi$-fields $e$ and $\nu$ are ordinary pointlike
fermion fields. Let us use the Dirac equation to handle them; we could
employ Weyl equations as well. The important feature is that the SBEG
procedure is thoroughly an on-shell construction:
\begin{equation}
\overrightarrow{\delslash} \psi = -im_\psi \psi,  \qquad
\bar\psi \overleftarrow{\delslash} = im_\psi \bar\psi.
\label{eq:Dirac-eqns} 
\end{equation}

String independence at this order demands that there be a
$Q^F_\mu(x,l)$ such that
\begin{equation}
d_l S_1^F(x,l) = \del^\mu Q^F_\mu(x,l).
\label{eq:SI-one} 
\end{equation}

\begin{prop} 
The string independence requirement~\eqref{eq:SI-one} can be satisfied
if and only if
\begin{alignat}{2}
c_1  &= i(m_e - m_\nu) b_1,  \qquad & c_3 &= 0,
\nonumber \\
c_2  &= i(m_\nu - m_e) b_1,  \qquad & c_4 &= 0,
\nonumber \\
\tilde c_1 &= i(m_e + m_\nu) \tilde b_1, \qquad
& \tilde c_3 &= 2i m_e \tilde b_3, 
\nonumber \\
\tilde c_2 &= i(m_\nu + m_e) \tilde b_1, \qquad
& \tilde c_4 &= 2i m_\nu \tilde b_4, 
\word{and} \tilde b_5 = 0.
\label{eq:good-coeffs} 
\end{alignat}
The corresponding $Q_1^{F\mu}$ is unique and is of the form
\begin{align}
Q_1^{F\mu} &:= g\bigl( b_1 w_- \bar e \ga^\mu \nu
+ \tilde b_1 w_- \bar e \ga^\mu \ga^5 \nu
+ b_1 w_+ \bar\nu \ga^\mu e + \tilde b_1 w_+ \bar\nu \ga^\mu \ga^5 e
\notag \\
&\qquad + b_3 w_Z \bar e \ga^\mu e
+ \tilde b_3 w_Z \bar e \ga^\mu \ga^5 e 
+ b_4 w_Z\bar\nu \ga^\mu \nu + \tilde b_4 w_Z\bar\nu \ga^\mu \ga^5 \nu
+ b_5 w_4 \bar e \ga^\mu e \bigr).
\label{eq:fermis-shadow} 
\end{align}
\end{prop}

Note that there are no restrictions at this stage on the set
$\{c_0, \tilde c_0, c_5, \tilde c_5\}$, since the corresponding
vertices are pointlike.

\begin{proof}
The string differential $d_l S_1^F$ with the Ansatz
\eqref{eq:fermis-wish-list} for $S_1^F$ is expressed with the help of
the form-valued fields defined in \eqref{eq:good-form}:
\begin{align*}
d_l S_1^F(x,l) 
&= g\bigl( b_1 \del_\mu w_- \bar e \ga^\mu \nu 
+ \tilde b_1 \del_\mu w_- \bar e \ga^\mu \ga^5 \nu
+ b_1 \del_\mu w_+ \bar\nu \ga^\mu e
+ \tilde b_1 \del_\mu w_+ \bar\nu \ga^\mu \ga^5 e
\\
&\qquad + b_3 \del_\mu w_Z \bar e \ga^\mu e
+ \tilde b_3 \del_\mu w_Z \bar e \ga^\mu \ga^5 e
+ b_4 \del_\mu w_Z \bar\nu \ga^\mu \nu
+ \tilde b_4 \del_\mu w_Z \bar\nu \ga^\mu \ga^5 \nu
\\
&\qquad + b_5 \del_\mu w_4 \bar e \ga^\mu e
+ \tilde b_5 \del_\mu w_4 \bar e \ga^\mu \ga^5 e
\\
&\qquad + c_1 w_- \bar e \nu + \tilde c_1 w_- \bar e \ga^5 \nu
+ c_2 w_+ \bar\nu e + \tilde c_2 w_+ \bar\nu \ga^5 e
\\
&\qquad + c_3 w_Z \bar e e + \tilde c_3 w_Z \bar e \ga^5 e
+ c_4 w_Z \bar\nu \nu + \tilde c_4 w_Z \bar\nu \ga^5 \nu \bigr).
\end{align*}
Using the Dirac equations~\eqref{eq:Dirac-eqns} and 
$\ga^5 \ga^\mu = -\ga^\mu \ga^5$ and \textit{defining} $Q_1^F$ as in
Eq.~\eqref{eq:fermis-shadow}, one finds:
\begin{align*}
d_l S_1^F(x,l) &= \del_\mu \bigl[ 
Q_1^{F\mu} + \tilde b^5 w_4 \bar e \ga^\mu \ga^5 e \bigr]
\\
&\quad + g\bigl[ (c_1 - i(m_e - m_\nu)b_1) w_- \bar e \nu 
+ (\tilde c_1 - i(m_e + m_\nu)\tilde b_1)  w_- \bar e \ga^5 \nu
\\
&\qquad + (c_2 - i(m_\nu - m_e)b_1) w_+ \bar\nu e
+ (\tilde c_2 - i(m_\nu + m_e) \tilde b_1) w_+ \bar\nu \ga^5 e
\\
&\qquad + (\tilde c_3 - 2im_e\tilde b_3) w_Z \bar e \ga^5 e 
+ (\tilde c_4 - 2im_\nu \tilde b_4) w_Z \bar\nu \ga^5 \nu
- 2im_e \tilde b_5 w_4 \bar e \ga^5 e
\\
&\qquad + c_3 w_Z \bar e e + c_4 w_Z \bar\nu \nu \bigl].
\end{align*}
The last four lines cannot be expressed as divergences, and by linear
independence of the cubic operators, the corresponding terms must
vanish separately. This implies the claims.
\end{proof}

Notice also that the argument for $\tilde b_5 = 0$ would have failed
if the electron were massless. Whereas the axial terms for massive
vector bosons in the original Ansatz have survived. They will keep
surviving, as we shall see.

It is pertinent to substitute expressions \eqref{eq:good-coeffs}
into~\eqref{eq:fermis-wish-list}, which we do now for convenience
later~on:
\begin{align}
S_1^F(x,l) &= g\bigl( b_1 W_{-\mu} \bar e \ga^\mu \nu 
+ \tilde b_1 W_{-\mu} \bar e \ga^\mu \ga^5 \nu
+ b_1 W_{+\mu} \bar\nu \ga^\mu e
+ \tilde b_1 W_{+\mu} \bar\nu \ga^\mu \ga^5 e
\notag \\
&\qquad + b_3 Z_\mu \bar e \ga^\mu e
+ \tilde b_3 Z_\mu \bar e \ga^\mu \ga^5 e
+ b_4 Z_\mu \bar\nu \ga^\mu \nu
+ \tilde b_4 Z_\mu \bar\nu \ga^\mu \ga^5 \nu
+ b_5 A_\mu \bar e \ga^\mu e
\notag \\
&\qquad + i(m_e - m_\nu) b_1 \phi_- \bar e \nu 
+ i(m_e + m_\nu) \tilde b_1 \phi_- \bar e \ga^5 \nu 
- i(m_e - m_\nu) b_1 \phi_+ \bar\nu e 
\notag \\
&\qquad + i(m_e + m_\nu)\tilde b_1 \phi_+ \bar\nu \ga^5 e
+ 2im_e \tilde b_3 \phi_Z \bar e \ga^5 e 
+ 2im_\nu \tilde b_4 \phi_Z \bar\nu \ga^5 \nu
\notag \\
&\qquad + c_0 \phi_4 \bar e e + \tilde c_0 \phi_4 \bar e \ga^5 e 
+ c_5 \phi_4 \bar\nu \nu + \tilde c_5 \phi_4 \bar\nu \ga^5 \nu \bigr).
\label{eq:fermis-mate} 
\end{align}

\section{Time-ordered products for tree graphs}
\label{sec:well-tempered}

Recall that the causal factorization~\eqref{eq:T2-factor} fixes the
time-ordered product $\T[S_1(x,l) S_1(x',l)]$ only outside the
set~$\sD$. The possible extensions across $\sD$ are restricted by the
requirement that the Wick expansion, valid outside $\sD$, hold
everywhere: we require that the time-ordered product of Wick
polynomials $U = U(x,l)$, $V' = V(x',l)$ satisfy
\begin{equation}  
\T [U V'] = \wick:U V': + \underbrace{\biggl[ \sum_{\vf,\chi'}
\wick:\frac{\del U}{\del\vf} \vev{\T \vf\,\chi'}\,
\frac{\del V'}{\del\chi'}: \biggr]}_{\displaystyle \T[UV']_\tree}
+\cdots+ \vev{\T [U V']}, 
\label{eq:Wick-expand} 
\end{equation}
where the sum in the brackets goes over all free fields, and we have
employed formal derivation within the Wick polynomial. The terms in
brackets are called the \textit{tree graphs}. Thereby the extension
problem is reduced to the extension of numerical distributions.

In particular, at the tree-graph level, it only remains to extend the
time-ordered two-point functions $\vev{\T \vf\,\chi'}$ of free fields.
One such extension is given by
\begin{equation}
\vev{\T_0 \vf(x,l)\,\chi(x',l)} := \frac{i}{(2\pi)^4} \int d^4p\,
\frac{e^{-i(p(x-x'))}}{p^2 - m^2 + i0}\, M^{\vf\chi}(p,l). 
\label{eq:pandemonium} 
\end{equation}
It has the nice feature that it preserves all off-shell relations
between the fields.%
\footnote{The string derivative $d_l$ fulfils the Leibniz rule 
with~$\T_0$ unconditionally. As long as no on-shell relations are
involved, $\del_\mu$ can be exchanged with $\T_0$ as well, e.g.:
$$
\del_\mu \vev{\T_0 A_\nu \chi'} - \del_\nu \vev{\T_0 A_\mu \chi'}
= \vev{\T_0 F_{\mu\nu} \chi'}.
$$}

If the scaling degree of the two-point function $\vev{\vf\,\chi'}$
with respect to $\sD$ and to the diagonal $\{x = x'\}$ is lower than
the respective codimensions $3$ and~$4$, then the time-ordered
two-point function is unique,
$\vev{\T \vf\,\chi'} = \vev{\T_0 \vf\,\chi'}$. Otherwise, it admits
the addition of a distribution with support on~$\sD$.

A look at the tables~\eqref{eq:two-pointers} shows that this happens
only in the cases $\vev{\T \del_\la A_\mu\,A_\ka'}$ and
$\vev{\T \del_\la A_\mu\,\del'_\ka w'}$. These have scaling degree~$3$
with respect to both $\sD$ and the diagonal $\{x = x'\}$, and
therefore admit a renormalization by adding a numerical distribution
supported on~$\sD$ and with the same scaling degree. Any such
distribution is of the form
\begin{equation}
\dl_l(x' - x) := \int_0^\infty ds\, \dl(x' - x - sl), 
\label{eq:delta-l} 
\end{equation}
multiplied by some well-behaved function $f(x' - x, l)$. Thus, in
these cases the most general two-point functions are
\begin{subequations}
\label{eq:bitter-tea} 
\begin{align}
\vev{\T \del_\la A_\mu\,A_\ka'} 
&= \vev{\T_0 \del_\la A_\mu\,A_\ka'} + c_{\la\mu\ka} \,\dl_l \,,
\label{eq:bitter-tea-light} 
\\ 
\vev{\T \del_\la A_\mu\,\del_\ka' w'} 
&= \vev{\T_0 \del_\la A_\mu\,\del_\ka' w'} + b_{\la\mu\ka} \,,
\label{eq:bitter-tea-dark} 
\end{align}
\end{subequations}
where $c_{\la\mu\ka}$ and $b_{\la\mu\ka}$ are some well-behaved
function and one-form respectively, as yet undetermined.

We now seek to enforce string independence of time-ordered products at
second order in the coupling constant. String independence at first
order~\eqref{eq:yet-unseen} plus the factorization~\eqref{eq:T2-factor}
imply that the relation
\begin{equation}
d_l \T[S_1(x,l) S_1(x',l)] 
= \del_\mu \T[Q_1^\mu(x,l) S_1(x',l)] 
+ \del'_\mu \T[S_1(x,l) Q_1^\mu(x',l)]
\label{eq:the-grail} 
\end{equation}
holds for all $(x - x', l)$ outside $\sD$. The
\textbf{string independence principle} forces us to require that this
relation be valid everywhere.
It turns out that this requirement fixes all coefficients in
\eqref{eq:fermis-wish-list}.

As advertised, to this end we shall only need to examine tree graphs
in~$S_2$. We reckon that the tree graph contribution to the
obstruction~\eqref{eq:the-grail} is given by
\begin{align} 
& \sum_{\vf,\chi'} \biggl[ d_l\frac{\del S_1}{\del \vf} 
\vev{\T \vf\chi'} \frac{\del S'_1}{\del\chi'}
+ \frac{\del S_1}{\del \vf}\, d_l\vev{\T \vf\chi'} 
\frac{\del S'_1}{\del\chi'}
+ \frac{\del S_1}{\del \vf} \vev{\T \vf\chi'} 
\,d_l \frac{\del S'_1}{\del\chi'} \biggr] 
\notag \\
&\qquad - \sum_{\psi,\chi'} \biggl(
\del_\mu \frac{\del Q^\mu}{\del\psi} \vev{\T \psi\chi'} 
+ \frac{\del Q^\mu}{\del\psi} \del_\mu \vev{\T \psi\,\chi'} \biggr)
\frac{\del S'_1}{\del\chi'} \; - [x \otto x'],
\label{eq:the-monster} 
\end{align}
where we have written $Q$ for~$Q_1$. This expression expands to
\begin{subequations}
\label{eq:the-ogre} 
\begin{align} 
&\sum_{\vf,\chi'} \biggl[ d_l\frac{\del S_1}{\del \vf}\,
\vev{\T \vf\chi'} \frac{\del S'_1}{\del\chi'}
+ \frac{\del S_1}{\del \vf} \vev{\T d_l \vf\chi'} 
\frac{\del S'_1}{\del\chi'}
+ \frac{\del S_1}{\del \vf} \vev{\T \vf d_l  \chi'} 
\frac{\del S'_1}{\del\chi'}
+ \frac{\del S_1}{\del \vf} \vev{\T \vf\chi'} 
\,d_l \frac{\del S'_1}{\del\chi'} \biggr]
\notag \\
&\qquad - \sum_{\psi,\chi'} \biggl(
\del_\mu \frac{\del Q^\mu}{\del\psi} \vev{\T \psi\chi'} 
+ \frac{\del Q^\mu}{\del\psi} \vev{\T \del_\mu\psi\,\chi'} \biggr)
\frac{\del S'_1}{\del\chi'} \; - [x \otto x']
\notag \\
&\qquad + \sum_{\vf,\chi'} \frac{\del S_1}{\del \vf} 
\bigl( d_l \vev{\T \vf \chi'} - \vev{\T d_l\vf \chi'}
- \vev{\T \vf d_l \chi'} \bigr) \frac{\del S'_1}{\del\chi'}
\label{eq:the-ogre-nasty} 
\\ 
&\qquad - \sum_{\psi,\chi'} \frac{\del Q^\mu}{\del\psi} \bigl(
\del_\mu \vev{\T \psi \chi'} - \vev{\T \del_\mu\psi \chi'} \bigr)
\frac{\del S'_1}{\del\chi'} \; - [x \otto x'].
\label{eq:the-ogre-horrid} 
\end{align}
\end{subequations}
The first, second, fifth and sixth terms reduce to a tree graph
contribution:
\begin{align}
& \sum_{\chi'} \biggl[ \sum_\vf \Bigl( 
d_l\frac{\del S_1}{\del \vf}\, \vev{\T \vf\chi'} 
+ \frac{\del S_1}{\del \vf} \vev{\T d_l \vf\chi'} \Bigr)
- \sum_\psi \Bigl(
\del_\mu \frac{\del Q^\mu}{\del\psi} \vev{\T \psi\chi'} 
+ \frac{\del Q^\mu}{\del\psi} \vev{\T \del_\mu\psi\,\chi'} \Bigr)
\biggr]
\frac{\del S'_1}{\del\chi'}
\notag \\
&\qquad = \T[(d_l S_1) S'_1]_\tree 
- \T[(\del_\mu Q^\mu) S'_1]_\tree \,,
\label{eq:the-secret-garden} 
\end{align}
which vanishes by construction; we refer to
Appendix~\ref{app:aliquid-obstat} for the proof of that equality. The 
other four terms in the first two summations vanish similarly. 

Thus, the whole expression~\eqref{eq:the-monster} reduces to the sum
\eqref{eq:the-ogre} of the last two lines above, which we may call the
``obstruction to string independence''.

We now seek to determine this quantity. Its vanishing, even admitting
the most general time-ordering prescription~$\T$, will provide the
correct couplings, and in the occasion chirality of the interaction of
the fermions with the massive intermediate vector bosons.

\medskip

We distinguish three types of $2$-point obstructions. For terms
$\vf, \chi$ in $S_1$ and $\psi, C^\mu$ in $Q_1^\mu$, we label them as 
follows: 
\begin{subequations}
\label{eq:obstat} 
\begin{align}
\wh\sO(\vf,\chi') 
&:= d_l \vev{\T \vf \chi'} - \vev{\T d_l\vf \chi'} 
- \vev{\T \vf\, d_l\chi'},
\label{eq:obstat-hat} 
\\  
\sO_\mu(\psi,\chi')
&:= \vev{\T \del_\mu \psi \chi'} - \del_\mu \vev{\T \psi \chi'},
\label{eq:obstat-mu} 
\\ 
\sO(C,\chi')
&:= \vev{\T \del_\mu C^\mu \chi'} - \del_\mu \vev{\T C^\mu \chi'}.
\label{eq:obstat-plain} 
\end{align}
\end{subequations}

Since the $\T_0$ ordering preserves all off-shell relations between
the fields, the first two types only occur for $\T \neq \T_0$. More
specifically, the only obstructions of these types that we meet are
\begin{subequations}
\label{eq:scandali} 
\begin{align} 
\wh\sO(F_{\mu\nu}, A'_\ka) &= d_l(c_{[\mu\nu]\ka}\,\dl_l), 
\label{eq:scandalum-hat} 
\\
\sO_\mu(w, F'_{\al\bt}) &= b_{[\al\bt]\mu} \,,
\label{eq:scandalum-mu} 
\end{align}
\end{subequations}
with skewsymmetrization 
$c_{[\mu\nu]\ka} \equiv c_{\mu\nu\ka} - c_{\nu\mu\ka}$ and similarly
for $b_{[\al\bt]\mu}$. These are numerical $1$-forms in the
$l$~variable. On the other hand, all obstructions of 
type~\eqref{eq:obstat-plain} are $0$-forms, since the only candidate
field $C^\mu$ for a $1$-form is $\del^\mu w$ -- but this does not
appear in~$Q_1^\mu$, see \eqref{eq:capeat-qui-potest}
and~\eqref{eq:higgs-shadow-terms}. We conclude that the terms in
\eqref{eq:the-ogre-horrid} which involve two-point obstructions of the
third type must cancel separately, i.e., cannot be cancelled by terms
involving the first two types of two-point obstructions.

\medskip

We now examine $2$-point obstructions of the third
type~\eqref{eq:obstat-plain}. First of all, there are two that vanish:
\begin{subequations}
\label{eq:nihil-obstat} 
\begin{align}
\sO(A,\phi') &:=
\vev{\T_0 \del_\mu A^\mu \phi'} - \del_\mu \vev{\T_0 A^\mu \phi'} = 0,
\label{eq:nihil-obstat-A-Phi} 
\\
\sO(\del_\la A,\phi') &:=
\vev{\T_0 \del_\mu \del_\la A^\mu \phi'} 
- \del_\mu \vev{\T_0 \del_\la A^\mu \phi'} = 0.
\label{eq:nihil-obstat-dA-Phi} 
\end{align}
\end{subequations}
Indeed, the left-hand side of~\eqref{eq:nihil-obstat-A-Phi} is 
$- m^2 \vev{\T_0 \phi \phi'} - \del_\mu \vev{\T_0 A^\mu \phi'}$,
which vanishes because
$$
\del_\mu \vev{\T_0 A^\mu\phi'} = \frac{-i}{(2\pi)^4} \int d^4p\,
\frac{e^{-i(p(x-x'))}}{p^2 - m^2 + i0} \equiv -i\,D_F(x - x')
= - m^2 \vev{\T_0 \phi \phi'},
$$
in view of~\eqref{eq:two-pointers-more}. Thus
\eqref{eq:nihil-obstat-A-Phi} holds; and a similar calculation
yields~\eqref{eq:nihil-obstat-dA-Phi}. Note that, by definition,
$$
D_F(x) := \frac{1}{(2\pi)^4} \int d^4p\, 
\frac{e^{-i(px)}}{p^2 - m^2 + i0} \,,  \word{so that}
(\square + m^2) D_F(x) = - \dl(x).
$$

Next, we consider
$$
\sO(A, A'_\ka) :=
\vev{\T_0 \del_\mu A^\mu A'_\ka} - \del_\mu \vev{\T_0 A^\mu A'_\ka}.
$$
Using \eqref{eq:two-pointers}, we get
$$
\sO(A, A'_\ka) = \frac{1}{(2\pi)^4} \int d^4p\,
\frac{e^{-i(p(x - x'))}}{p^2 - m^2 + i0}\,
\frac{(m^2 - p^2)l_\ka}{(pl)} 
= - \frac{l_\ka}{(2\pi)^4} \int d^4p\, \frac{e^{-i(p(x - x'))}}{(pl)}\,.
$$
On bringing in the distributions
$1/(pl) = -i \int_0^\infty ds\, e^{is(pl)}$ and $\dl_l$
of~\eqref{eq:delta-l}, we may rewrite the obstruction as
\begin{equation}
\sO(A, A'_\ka) = \frac{il_\ka}{(2\pi)^4} \int_0^\infty ds\, 
\int d^4p\, e^{-i(p(x - x' - sl))} = il_\ka \,\dl_l(x - x').
\label{eq:obs-A-A} 
\end{equation}

We next determine
\begin{align*}
\sO(\del\phi, A'_\ka)
&:= \vev{\T_0 \del_\mu \del^\mu\phi A'_\ka}
- \del_\mu \vev{\T_0 \del^\mu\phi A'_\ka}
\\
&= - (\square + m^2) \vev{\T_0 \phi A'_\ka} 
= - \frac{1}{(2\pi)^4} \int d^4p\, e^{-i(p(x-x'))}\,
\frac{l_\ka}{(pl)} = i l_\ka \,\dl_l \,.
\end{align*}
Since $\sO$ is bilinear in its arguments, this yields a useful result:
$\sO(A - \del\phi, A'_\ka) = 0$. Likewise,
$$
\sO(\del A_\la, \phi') := \vev{\T_0 \del_\mu \del^\mu A_\la \phi'}
- \del_\mu \vev{\T_0 \del^\mu A_\la \phi'} 
= - (\square + m^2) \vev{\T_0 A_\la \phi'}
= - i l_\la \,\dl_l \,. 
$$

\medskip

We now tackle the obstruction $\sO(\del_\la A, A'_\ka)$, which
involves $\vev{\T \del_\la A_\mu A'_\ka}$ that is not unique but
admits the renormalization~\eqref{eq:bitter-tea-light}. To wit,
\begin{align}
\sO(\del_\la A, A'_\ka)
&:= \vev{\T\, \del_\mu \del_\la A^\mu A'_\ka} 
- \del_\mu \vev{\T\, \del_\la A^\mu A'_\ka}
\notag \\
&= \del_\la \bigl( - m^2 \vev{\T_0 \phi A'_\ka} 
- \del^\mu \vev{\T_0 A_\mu A'_\ka} \bigr) 
- \del^\mu \bigl( c_{\la\mu\ka} \,\dl_l \bigr)
\notag \\
&= i l_\ka\,\del_\la \dl_l 
- \del^\mu \bigl( c_{\la\mu\ka}\,\dl_l \bigr).
\label{eq:obs-dA-A} 
\end{align}
Next, we find, using \eqref{eq:two-pointers-AA}
and~\eqref{eq:obs-A-A}, that
\begin{align} 
\sO(\del A_\la, A'_\ka)
&:= \vev{\T\, \del_\mu \del^\mu A_\la  A'_\ka}
-  \del_\mu \vev{\T \, \del^\mu A_\la A'_\ka} 
\notag \\
&= -(\square + m^2) \vev{\T_0 A_\la A'_\ka}
- \del^\mu(c_{\mu\la\ka} \,\dl_l)
\notag \\
&= -i g_{\la\ka} \,\dl + i(l_\la \del_\ka + l_\ka \del_\la)\,\dl_l 
- \del^\mu(c_{\mu\la\ka} \,\dl_l).
\label{eq:obs-dA-A'} 
\end{align}
On subtracting \eqref{eq:obs-dA-A} from~\eqref{eq:obs-dA-A'}, we
arrive at
$$
\sO(F_{\8\la}, A'_\ka) \equiv
\sO(\del_\la A - \del A_\la, A'_\ka) 
= -i g_{\la\ka} \,\dl + i l_\la \del_\ka \,\dl_l 
+ \del^\mu(c_{[\la\mu]\ka} \,\dl_l).
$$

Finally, we take note of
\begin{align*} 
\sO(\del\phi_a, \phi'_a) 
&:= \vev{\T_0 \del_\mu \del^\mu \phi_a \phi'_a}
-  \del_\mu \vev{\T_0 \del^\mu \phi_a \phi'_a} 
= \frac{i}{m_a^2}\,\dl  \word{for}  a = 1,2,3;
\\  
\sO(\del\phi_4, \phi'_4) 
&:= \vev{\T_0 \del_\mu \del^\mu \phi_4 \phi'_4}
-  \del_\mu \vev{\T_0 \del^\mu \phi_4 \phi'_4} = i\,\dl.
\end{align*}

\medskip

To sum up: the obstructions of the third bosonic type are: 
\begin{align}
\sO(A, \phi') &= 0,
&\sO(A, A'_\ka) &= il_\ka \,\dl_l \,, 
\notag \\
\sO(\del_\la A, \phi') &= 0,
& \sO(\del\phi, A'_\ka) &= i l_\ka \,\dl_l \,,
\notag \\
\sO(A - \del\phi, A'_\ka) &= 0,  
& \sO(\del A_\la, \phi') &= - i l_\la \,\dl_l \,,
\label{eq:tabula-rasa} 
\\
\sO(\del\phi_a, \phi'_a) &= (i/m_a^2)\,\dl \,,
& \sO(\del_\la A, A'_\ka) &= i l_\ka\,\del_\la \dl_l 
- \del^\mu \bigl( c_{\la\mu\ka}\,\dl_l \bigr),
\notag \\
\sO(\del\phi_4, \phi'_4) &= i\,\dl \,,
& \sO(\del A_\la, A'_\ka)
&= -i g_{\la\ka} \,\dl + i(l_\la \del_\ka + l_\ka \del_\la)\,\dl_l 
- \del^\mu \bigl( c_{\mu\la\ka}\,\dl_l \bigr),
\notag \\
&& \sO(F_{\8\la}, A'_\ka) 
&= -i g_{\la\ka} \,\dl + i l_\la \del_\ka \,\dl_l 
+ \del^\mu(c_{[\la\mu]\ka} \,\dl_l).
\notag 
\end{align}

The fermionic obstructions, which do not involve stringlike fields,
are much simpler. They are of two kinds, where $\psi$, $\psi'$ denote
two fermions of the same type:
\begin{align}
\sO(\ga \psi, \bar \psi')
&:= \vev{\T_0 \ga^\mu \del_\mu \psi \,\bar \psi'}
-  \ga^\mu \del_\mu \vev{\T_0 \psi \,\bar \psi'} = -\dl,
\notag \\
\sO(\psi', \bar \psi \ga)
&:= \vev{\T_0 \psi' \,\del_\mu \bar \psi} \ga^\mu 
-  \del_\mu \vev{\T_0 \psi' \,\bar \psi} \ga^\mu  = +\dl.
\label{eq:delta-glitch} 
\end{align}
Indeed, using~\eqref{eq:Dirac-eqns}, we obtain
$$
\sO(\ga \psi, \bar \psi') 
= - (\delslash + im_\psi) \vev{\T_0 \psi\,\bar \psi'}
= -i(\delslash + im_\psi) S^F(x - x') = - \dl(x - x'),
$$
and the second case follows similarly.

\section{Computing the second-order constraints}
\label{sec:at-laeva-malorum}

\textit{A priori}, in equation \eqref{eq:the-grail} there may be three
kinds of contractions pertinent to our problem of the type
\eqref{eq:the-ogre-horrid}, coming from the \textit{crossing} of the
respective bosonic and fermionic couplings $S_1^B$ and $S_1^F$ with
the $Q_1^B$ and $Q_1^F$ vector operators. These crossings contain
information about the fermionic vertices. Happily, the bosonic
interaction set $S_1^B$ and the fermionic $Q_1^F$-vertex turn out an
inert combination, because there are no obstructions involving the
form-valued fields~$w_a$.

Our goal in this section is to determine the couplings, as far as
possible, from the vanishing of obstructions
in~\eqref{eq:the-ogre-horrid} of the third
type~\eqref{eq:obstat-plain} -- which have to vanish separately from
the other two types as remarked after Eq.~\eqref{eq:scandali}.
Firstly, we seek the $\tilde b_3$ and~$\tilde b_4$ coefficients of the
$Z$-boson, which are determined together with the higgs couplings
$c_0$ and~$c_5$. Secondly, we shall be able to determine the quotient
$b_1/\tilde b_1$, thereby obtaining chirality of the charged boson
interactions in the SM; the value of $b_1$ is trivially determined
afterwards. Thirdly, we shall look for the electromagnetic coupling
$b_5$. At the end, we find the missing terms for the neutral current
and show vanishing of the other higgs couplings.

In what follows, we consider two types of crossings. The first
involves a $Q_1^B$ vector $\psi$, namely a summand taken from the
formulas \eqref{eq:capeat-qui-potest}
and~\eqref{eq:higgs-shadow-terms}, and a $S_1^F$ coupling $\chi'$ that
is a summand of~\eqref{eq:fermis-mate}; these we call
$(Q_1^B,S_1^F)$-type crossings. The second type pairs a $Q_1^F$ vector
summand $\psi$ of~\eqref{eq:fermis-shadow} with a term $\chi'$
in~\eqref{eq:fermis-mate}; these will be $(Q_1^F,S_1^F)$-type
crossings. (The possible fermionic crossings are listed in
Appendix~\ref{app:cross-purposes}.) Each such crossing yields a single
term in the total obstruction \eqref{eq:the-ogre-horrid}, consisting of a
$2$-point obstruction combined with certain (Wick) products of fields.
Different individual crossings may, and will, turn out to have the
same field content -- which give opportunities for cancellation of
their obstructions.

For convenience and readability, we shall omit the factor $g^2$ in
all crossings in this section, reinstating it in the final result.

\subsection{Step 1: impact of higgs couplings}

\begin{lema} 
\label{lm:first-higgs}
The crossings with field content $w_Z(x,l) \phi_Z(x,l) \bar e(x) e(x)$
yield no obstruction to string independence, if and only if the higgs
and $Z$-boson coupling coefficients $c_0$ and $\tilde b_3$ are related
as follows:
\begin{equation}
c_0 = \frac{8 \tilde b_3^2\, m_e \cos^2\Th}{m_W}\,.
\label{eq:first-higgs} 
\end{equation}
\end{lema}

\begin{proof}
One such crossing, of the $(Q_1^B,S_1^F)$-type, arises from the 
term $-\frac{1}{2\cos\Th}\, m_Z \,\del_\mu\phi_4\, \phi_Z w_Z$
in~\eqref{eq:higgs-shadow-terms} with the term $c_0 \phi_4 \bar e e$
in~\eqref{eq:fermis-mate}. From the table \eqref{eq:tabula-rasa},
this contributes to the total obstruction the term:
$$
- i c_0 \frac{m_Z}{\cos\Th}\, w_Z(x,l) \phi_Z(x,l) \bar e(x) e(x)
 \,\dl(x - x').
$$
A factor of~$2$ comes from appending the identical second contribution
in \eqref{eq:the-grail}; we do likewise from now on without further
notice.

On the other hand, there is a crossing of type $(Q_1^F,S_1^F)$,
matching $\tilde b_3 w_Z \bar e \ga^\mu \ga^5 e$
in~\eqref{eq:fermis-shadow} and
$2i m_e \tilde b_3 \phi_Z \bar e \ga^5 e$ in~\eqref{eq:fermis-mate}.
Here there are two 
$\bar e$-$e$ contractions of equal value, see the
table~\eqref{eq:bearing-the-cross}, for a total contribution of
$$
8i m_e \tilde b_3^2\, w_Z(x,l) \phi_Z(x,l) \bar e(x) e(x)
\,\dl(x - x').
$$
String independence therefore demands cancellation of the last two
expressions; since there are no more crossings with this field
content, this yields~\eqref{eq:first-higgs}.
\end{proof}

\begin{lema} 
\label{lm:second-higgs}
The crossings with field content 
$w_Z(x,l) \phi_4(x) \bar e(x) \ga^5 e(x)$
yield no obstruction to string independence, if and only if 
$c_0 = m_e/2 m_W$. Hence
\begin{equation}
\tilde b_3 = \pm \frac{1}{4\cos\Th} =: \eps_1 \frac{1}{4\cos\Th}
\label{eq:neutral-zone} 
\end{equation}
where the sign $\eps_1 = \pm 1$ is yet to be determined.
\end{lema}

\begin{proof}
There is one crossing of type $(Q_1^B,S_1^F)$, of 
$\frac{1}{2\cos\Th} m_Z w_Z \phi_4 \,\del_\mu\phi_Z$ from
\eqref{eq:higgs-shadow-terms} with the term
$2i m_e \tilde b_3 \phi_Z \bar e \ga^5 e$
from~\eqref{eq:fermis-shadow}. For this one, \eqref{eq:tabula-rasa}
yields
$$
-2 \tilde b_3 \frac{m_e}{m_W}\,
w_Z(x,l) \phi_4(x) \bar e(x) \ga^5 e(x) \,\dl(x - x').
$$
Now there are two relevant $(Q_1^F,S_1^F)$-type crossings:
$\tilde b_3 w_Z \bar e \ga^\mu \ga^5 e$ with $c_0 \phi_4 \bar e e$ and 
$b_3 w_Z \bar e \ga^\mu e$ with $\tilde c_0 \phi_4 \bar e \ga^5 e$.
The second vanishes -- see~\eqref{eq:bearing-the-cross} again -- and
the first yields
$$
4 \tilde b_3 c_0\, w_Z(x,l) \phi_4(x) \bar e(x) \ga^5 e(x)
\,\dl(x - x').
$$

Cancellation of these crossings requires $c_0 = m_e/2 m_W$, as 
claimed. Comparing that with the relation~\eqref{eq:first-higgs},
we arrive at $\tilde b_3^2 = 1/(16 \cos^2\Th)$, and
\eqref{eq:neutral-zone} follows.
\end{proof}

\begin{lema} 
\label{lm:third-higgs}
The vanishing of obstructions implies similar relations between the
higgs and $Z$-boson coupling coefficients $c_5$ and $\tilde b_4$:
$$
c_5 = \frac{8 \tilde b_4^2\, m_\nu \cos^2\Th}{m_W} 
= \frac{m_\nu}{2 m_W}
$$
and thereby leads to a determination of $\tilde b_4$ with another 
unspecified sign~$\eps_2$:
\begin{equation}
\tilde b_4 = \pm \frac{1}{4\cos\Th} =: \eps_2 \frac{1}{4\cos\Th}\,.
\label{eq:twilight-zone} 
\end{equation}
\end{lema}

\begin{proof}
In much the same way as before, we look now for crossings of either 
type with field content $w_Z(x,l) \phi_Z(x,l) \bar\nu(x) \nu(x)$.
There are just two of these: 
$-\frac{1}{2\cos\Th} m_Z w_Z \,\del_\mu\phi_4\, \phi_Z$ with
$c_5 \phi_4 \bar\nu \nu$ and 
$\tilde b_4 w_Z \bar\nu \ga^\mu \ga^5 \nu$ with
$2i m_\nu \tilde b_4 \phi_Z \bar\nu \ga^5 \nu$. These cancel provided
that $c_5$ and $\tilde b_4$ satisfy the first relation above.

On the other hand, the field content
$w_Z(x,l) \phi_4(x) \bar\nu(x) \ga^5 \nu(x)$ can arise from four 
crossings: $\frac{1}{2\cos\Th} m_Z w_Z \phi_4(\del_\mu\phi_Z - Z_\mu)$
with both $2i m_\nu \tilde b_4 \phi_Z \bar\nu \ga^5 \nu$ and
$\tilde b_4 Z_\ka \bar\nu \ga^\ka \ga^5 \nu$; and moreover the
$(Q_1^F,S_1^F)$-type ones $\tilde b_4 w_Z \bar\nu \ga^\mu \ga^5 \nu$ 
with $c_5 \phi_4 \bar\nu \nu$, and $b_4 w_Z \bar\nu \ga^\mu \nu$ with
$\tilde c_5 \phi_4 \bar\nu \ga^5 \nu$. The second and fourth of these
again vanish. Cancellation of the first and third leads to
$c_5 = m_\nu/2 m_W$; and $\tilde b_4^2 = 1/(16 \cos^2\Th)$ follows at 
once.
\end{proof}

Note that the higgs couplings $c_0$ and $c_5$ come out respectively 
proportional to the electron and neutrino masses, with the same 
proportionality constant -- as it should be.%
\footnote{We have left aside the possibility that $\tilde b_3$, $c_0$,
$\tilde b_4$ and~$c_5$ all vanish; this will soon be refuted.}

\subsection{Step 2: the road to chirality}

The signs $\eps_1$ and $\eps_2$ turn out to be related. This is the 
main step in the proof.

\begin{lema} 
\label{lm:auspicious-signs}
The coefficients $\tilde b_3$ and $\tilde b_4$ have opposite signs:
$\eps_2 = - \eps_1$.
\end{lema}

\begin{proof}
Consider together obstructions with field contents
$w_- W_{+\ka} \bar e \ga^\ka \ga^5 e$ and
$w_+ W_{-\ka} \bar e \ga^\ka \ga^5 e$. They may come from crossings
of type $(Q_1^F,S_1^F)$:
\begin{align*}
b_1 w_- \bar e \ga^\mu \nu 
& \word{with} \tilde b_1 W_{+\ka} \bar\nu \ga^\ka \ga^5 e 
& \text{and}\quad \tilde b_1 w_- \bar e \ga^\mu\ga^5 \nu
& \word{with} b_1 W_{+\ka} \bar\nu \ga^\ka e,
\\
b_1 w_+ \bar\nu \ga^\mu e 
& \word{with} \tilde b_1 W_{-\ka} \bar e \ga^\ka \ga^5 \nu 
& \text{and}\quad \tilde b_1 w_+ \bar\nu \ga^\mu\ga^5 e
& \word{with} b_1 W_{-\ka} \bar e \ga^\ka \nu.
\end{align*}
Each line gives rise to two identical obstructions, with total value
$$
-4 b_1 \tilde b_1 (w_- W_{+\ka} - w_+ W_{-\ka})
\bar e \ga^\ka \ga^5 e \,\dl(x - x').
$$

Such a term also arises from the $(Q_1^B,S_1^F)$-type crossing of the
term $i \cos\Th\, (w_- W_+^\la - w_+ W_-^\la) \, F^Z_{\mu\la}$
in~\eqref{eq:capeat-qui-potest} with
$\tilde b_3 Z_\ka \bar e \ga^\ka \ga^5 e$. As we saw in
Section~\ref{sec:well-tempered}, this is a ``dangerous'' crossing,
yielding
\begin{align*}
& 2 \tilde b_3 \cos\Th\, (w_- W_{+\ka} - w_+ W_{-\ka})
\bar e \ga^\ka \ga^5 e \,\dl(x - x')
\\
&+ 2i \tilde b_3 \cos\Th\, (w_- W_+^\la - w_+ W_-^\la)
\bar e \ga^\ka \ga^5 e
\,\del^\mu\bigl( c_{[\la\mu]\ka}\,\dl_l(x - x') \bigr).
\end{align*}
The term $i l_\la \,\del_\ka \dl_l$ in $\sO(F^Z_{\8\la}, Z'_\ka)$
does not contribute, since $l_\la W_\pm^\la = 0$ by transversality
(see Section~\ref{sec:on-a-string}). We obtain, in all:
\begin{align*}
& (2 \tilde b_3 \cos\Th - 4 b_1 \tilde b_1) 
(w_- W_{+\ka} - w_+ W_{-\ka}) \bar e \ga^\ka \ga^5 e \,\dl(x - x')
\\
&- 2i \tilde b_3 \cos\Th\, 
(w_- W_+^\la - w_+ W_-^\la) \bar e \ga^\ka \ga^5 e 
\,\del^\mu\bigl( c_{[\la\mu]\ka}\,\dl_l(x - x') \bigr).
\end{align*}
Here string independence dictates that $c_{[\la\mu]\ka} = 0$.%
\footnote{This implies that all two-point obstructions of the first
type~\eqref{eq:obstat-hat} also vanish, see~\eqref{eq:scandalum-hat}.
Those of the second type~\eqref{eq:obstat-mu} can be freely set to
zero, since they involve the up-to-now free parameters
$b_{\al\bt\mu}$, see~\eqref{eq:scandalum-mu}.}
The end result is
\begin{subequations}
\label{eq:almost-there} 
\begin{equation}
2 b_1 \tilde b_1 = \tilde b_3 \cos\Th.
\label{eq:almost-there-e} 
\end{equation}
A completely parallel computation, for obstructions with the field 
contents
$w_\mp W_{\pm\ka} \bar\nu \ga^\ka \ga^5 \nu$, gives the relation
\begin{equation}
2 b_1 \tilde b_1 = - \tilde b_4 \cos\Th.
\label{eq:almost-there-nu} 
\end{equation}
\end{subequations}
In view of~\eqref{eq:neutral-zone} and~\eqref{eq:twilight-zone},
this says that $\eps_2 = -\eps_1$.
\end{proof}

\begin{corl} 
\label{cr:auspicious-signs}
The interactions with fermions of the charged vector bosons must be
fully chiral, because $\tilde b_1 = \eps_1 b_1$.
\end{corl}

\begin{proof}
We now observe that $w_- \phi_Z \bar e \nu$ is produced either by the
term from \eqref{eq:bosons-shadow} of the form
$\ihalf m_W^2 \sec\Th\, w_- \,\del_\mu\phi_+ \,\phi_Z$, crossed with
$i(m_e - m_\nu) b_1 \phi_- \bar e \nu$ from \eqref{eq:fermis-mate}; or
by purely fermionic crossings, between
$\tilde b_1 w_- \bar e \ga^\mu \ga^5 \nu$ and the terms
$2im_e \tilde b_3 \phi_Z \bar e \ga^5 e
+ 2im_\nu \tilde b_4 \phi_Z \bar\nu \ga^5 \nu$. This, together with
\eqref{eq:neutral-zone} and~\eqref{eq:twilight-zone}, leads to
$$
i(m_e - m_\nu) b_1
= 2\tilde b_1(2i m_e \tilde b_3 + 2i m_\nu \tilde b_4) \cos\Th
= i(m_e - m_\nu) \eps_1 \tilde b_1 
$$
and the relation $\tilde b_1 = \eps_1 b_1$ follows.
\end{proof}

Of course, this procedure cannot tell us whether $\eps_1 = +1$ or
$\eps_1 = -1$. The second of these appears to be Nature's decision.

Equations~\eqref{eq:almost-there} now dictate that
$b_1^2 = \tilde b_1^2 = 1/8$. This determines $b_1$, up to a sign; we
choose $b_1 = - 1/2\sqrt{2}$.

\medskip

Observe that the proof of chirality requires the presence of a higgs,
at the level of tree graphs.
(Indeed, were $\tilde b_3 = 0$ or $\tilde b_4 = 0$,
it would follow that $b_1 = \tilde b_1 = 0$ too, and the whole term
$S_1^F$ would vanish. Thus none of these coefficients are zero, and
\eqref{eq:neutral-zone} is confirmed, with $c_0 \neq 0$ and 
$c_5 \neq 0$ as well.)
There are several consistency cases for the scalar particle of the
Standard Model. But it is hard to think of a simpler one. (We owe this
remark to Alejandro Ibarra.)

\subsection{Step 3: electric charge}

The coefficient $\mathrm{e} = g b_5$ of the coupling
$A_\mu \bar e \ga^\mu e$ in~\eqref{eq:fermis-mate} is just the 
electric charge. An important tenet of electroweak 
theory~\cite{Nagashima13} is that $\mathrm{e} = g \sin\Th$, with
$\Th$ being the Weinberg angle.

\begin{lema} 
\label{lm:light-brigade-charge}
The relation $g b_5 = g \sin\Th$ holds.
\end{lema}

\begin{proof}
Consider the term $-i\sin\Th\, w_- \,A^\la \,F_{\mu\la}^+$ in
\eqref{eq:capeat-qui-potest}, crossed with the term
$b_1 W_{-\ka} \bar e \ga^\ka \nu$ in \eqref{eq:fermis-mate}; and the
crossing of $b_1 w_- \bar e \ga^\mu \nu$ with 
$b_5 A_\ka \bar e \ga^\ka e$. These are the only terms yielding the
field content $w_- A_\ka \bar e \ga^\ka \nu$. The total obstruction is
$$
(2 b_1 b_5 - 2 b_1 \sin\Th)
w_-(x,l) A_\ka(x,l) \bar e(x) \ga^\ka \nu(x) \,\dl(x - x').
$$
This vanishes if and only if $b_5 = \sin\Th$.
\end{proof}

The case could also have been made from the crossings with field 
content $w_+ A_\ka \bar\nu \ga^\ka e$, \textit{mutatis mutandis}.

\subsection{Step 4: mopping up}

We still have to determine the couplings $b_3$ and $b_4$ for the
neutral current. For that, we seek first the contributions with
content $w_- W_{+\ka} \bar e \ga^\ka e$. The crossings are of four
classes:
\begin{align*}
i \sin\Th\, w_- W_+^\la F_{\mu\la}
&\word{with} b_5 A_\ka \bar e \ga^\ka e,
\\
i \cos\Th\, w_- W_+^\la F^Z_{\mu\la} 
&\word{with} b_3 Z_\ka \bar e \ga^\ka e,
\\
b_1 w_- \bar e \ga^\mu \nu 
&\word{with} b_1 W_{+\ka} \bar\nu \ga^\ka e,
\\
\tilde b_1 w_- \bar e \ga^\mu \ga^5 \nu  
&\word{with} \tilde b_1 W_{+\ka} \bar\nu \ga^\ka \ga^5 e.
\end{align*}
The cancellation of the total obstruction now entails
$$
b_3 \cos\Th + \sin^2\Th = b_1^2 + \tilde b_1^2 = \frac{1}{4}\,,
\word{that is,}
b_3 = \frac{1}{4\cos\Th} - \frac{\sin^2\Th}{\cos\Th}\,.
$$
Similarly, from the crossing of $i \cos\Th\, w_- W_+^\la F^Z_{\mu\la}$
with $b_4 Z_\ka \bar\nu \ga^\ka \nu$, and the same fermionic 
terms as before, the contributions with content
$w_- W_{+\ka} \bar\nu \ga^\ka \nu$ cancel only if
$$
b_4 \cos\Th = - b_1^2 - \tilde b_1^2 = - \frac{1}{4}\,,
\word{and thus}  b_4 = - \frac{1}{4\cos\Th} \,.
$$
The expected result of the neutral current containing a right-handed
component has been obtained.

Finally, crossing the term 
$-\half m_Z \sec\Th w_Z \phi_Z \,\del_\mu\phi_4$
in~\eqref{eq:higgs-shadow-terms} with the terms
$\tilde c_0 \phi_4 \bar e \ga^5 e$ and
$\tilde c_5 \phi_4 \bar\nu \ga^5 \nu$ of~\eqref{eq:fermis-mate} 
gives rise to terms with content $w_Z \phi_Z \bar e \ga^5 e$
and $w_Z \phi_Z \bar\nu \ga^5 \nu$, respectively. 

The  crossings of $b_3 w_Z \bar e \ga^\mu e$ with 
$2i m_e \tilde b_3 \phi_Z \bar e \ga^5 e$ and
$b_4 w_Z \bar\nu \ga^\mu \nu$ with 
$2i m_\nu \tilde b_4 \phi_Z \bar\nu \ga^5 \nu$, respectively, vanish
of their own accord: see the table~\eqref{eq:bearing-the-cross}.
Therefore, they cannot cancel the former crossings, and so
$\tilde c_0 = \tilde c_5 = 0$ must hold. That is to say, the couplings
of the higgs \textit{are not chiral}.

\medskip

In conclusion, we exhibit the leptonic couplings (for one family) of
the SM, as derived from string independence. For definiteness, we take
$\eps_1 = -1$, which is the experimental fact. Here, then, is the
chirality theorem in full.

\begin{thm} 
\label{th:EW-chirality}
The couplings of electroweak bosons to the fermion sector of the
Standard Model are fully determined from string independence and the
choice of sign $\eps_1 = -1$. Given that choice, the absence of
obstructions to string independence, at tree level up to second order,
entails that:
\begin{align}
S_1^F = g \biggl\{ 
& - \frac{1}{2\sqrt{2}}\, W_{-\mu} \bar e \ga^\mu (1 - \ga^5) \nu
- \frac{1}{2\sqrt{2}}\, W_{+\mu} \bar\nu \ga^\mu (1 - \ga^5) e
+ \frac{1 - 4\sin^2\Th}{4\cos\Th}\, Z_\mu \bar e \ga^\mu e
\notag \\
& - \frac{1}{4\cos\Th}\, Z_\mu \bar e \ga^\mu \ga^5 e
- \frac{1}{4\cos\Th}\, Z_\mu \bar\nu \ga^\mu (1 - \ga^5) \nu
+ \sin\Th\, A_\mu \bar e \ga^\mu e
\notag \\
& + i \frac{m_e - m_\nu}{2\sqrt{2}}\,
(\phi_- \bar e \nu - \phi_+ \bar\nu e)
- i \frac{m_e + m_\nu}{2\sqrt{2}}\, 
(\phi_- \bar e \ga^5 \nu + \phi_+ \bar\nu \ga^5 e)
\notag \\
& - i \frac{m_e}{2\cos\Th}\, \phi_Z \bar e \ga^5 e 
+ i \frac{m_\nu}{2\cos\Th}\, \phi_Z \bar\nu \ga^5 \nu
+ \frac{m_e}{2m_W}\, \phi_4 \bar e e
+ \frac{m_\nu}{2m_W} \phi_4 \bar\nu \nu \biggr\}.
\label{eq:ave-atque-vale} 
\end{align}
\end{thm}

Amazingly, this differs from what is known from the standard treatment
by little more than a divergence.

\begin{schol} 
One can write 
$S_1^F = S_1^{F,\pt} + (\del V)$, where $S_1^{F,\pt}$ is almost
pointlike,
\begin{align}
S_1^{F,\pt} = g \biggl\{ 
& - \frac{1}{2\sqrt{2}}\, W^\pt_{-\mu} \bar e \ga^\mu (1 - \ga^5) \nu
- \frac{1}{2\sqrt{2}}\, W^\pt_{+\mu} \bar\nu \ga^\mu (1 - \ga^5) e
+ \frac{1 - 4\sin^2\Th}{4\cos\Th}\, Z^\pt_\mu \bar e \ga^\mu e
\notag \\
& - \frac{1}{4\cos\Th}\, Z^\pt_\mu \bar e \ga^\mu \ga^5 e
- \frac{1}{4\cos\Th}\, Z^\pt_\mu \bar\nu \ga^\mu (1 - \ga^5) \nu
+ \sin\Th\, A_\mu \bar e \ga^\mu e
\notag \\
& + \frac{m_e}{2m_W}\, \phi_4 \bar e e
+ \frac{m_\nu}{2m_W} \phi_4 \bar\nu \nu \biggr\};
\label{eq:whats-the-point} 
\end{align}
where $V^\mu$ is given by
\begin{align*}
V^\mu = g \biggl\{ 
& - \frac{1}{2\sqrt{2}}\, \phi_- \bar e \ga^\mu (1 - \ga^5) \nu
- \frac{1}{2\sqrt{2}}\, \phi_+ \bar\nu \ga^\mu (1 - \ga^5) e
+ \frac{1 - 4\sin^2\Th}{4\cos\Th}\, \phi_Z \bar e \ga^\mu e
\\
& - \frac{1}{4\cos\Th}\, \phi_Z \bar e \ga^\mu \ga^5 e
- \frac{1}{4\cos\Th}\, \phi_Z \bar\nu \ga^\mu (1 - \ga^5) \nu
\biggr\}.
\end{align*}
That is to say, the divergence of the expression $V$ sweeps away the
escort fields.
\end{schol}

We wrote ``almost pointlike'' because the fields in
\eqref{eq:whats-the-point} are pointlike, except for the photon field
$A_\mu$, which remains stringlike -- for the good reason that $W_\pm$
and~$Z$ can be lodged in a Hilbert space, whereas $A$ cannot.
Incidentally, this causes the interacting electron field to be
string-localized, thus making direct contact with the early literature
on stringlike fields \cite{Mandelstam62,Steinmann84}. A key
observation is that $(\del V)$ is \textit{not renormalizable} by power
counting, whereas $(\del Q)$~is.

\medskip

We rest our case. The only way to disprove it would be to find an
inconsistency coming from crossings not discussed so far. To verify
that this does not happen is a routine, if utterly tedious, exercise.

A last remark is in order. In the stringlike version of electroweak
theory, the eventual need of ``renormalizing'' the original
time-ordered product~$\T_0$, as in~\eqref{eq:bitter-tea-light},
arises. We only found that the skewsymmetric part of~$c_{\la\mu\ka}$
in that formula must vanish. Whether or not the theory requires a
time-ordered product different from~$\T_0$ remains an open question.

\section{Conclusion and outlook}
\label{sec:our-revels-now-are-ended}

To repeat ourselves: interactions of quanta should spring from a
simple underlying principle. Gauge field theory has played this
unifying role so far. That flows from the embarrassing clash of the
positivity axioms of Quantum Mechanics with the convenient description
of electromagnetic and other forces in terms of potentials. Not
unreasonably, the difficulty was elevated into a principle, and one
that put geometry in the saddle. The resulting top-down approach, with
the need of ``quantizing'' the Lagrangian description, has ridden us
(without much mercy) for many a year. It should be recognized,
however, that the gauge-plus-BRST-invariance framework is just a very
useful theoretical \textit{technology} to grapple with elementary
particle physics problems. Other theoretical technologies can, and
sometimes are and should be, used to address them. Stringlike field
theory is but one of those. With the early dividends that the
mentioned clash fades away, and unbounded-helicity particles take
their due place among quantum fields~\cite{MundSY04}.

To be sure, the extra variable complicates renormalized perturbation
theory and the proof of renormalizability of physical models in
general. Notwithstanding, the string independence principle becomes a
powerful guide to interacting models. Internal symmetries are shown as
consequences of quantum mechanics in the presence of Lorentz symmetry,
and a bottom-up construction of the string-local equivalent for
self-interaction of the Yang--Mills type ensues~\cite{WienNotes}.
Fortunately, as with the chirality theorem itself, all that and more
requires only construction of time-ordered products associated with
tree graphs.%
\footnote{There is nothing much new in this: in the seventies it was
generally understood that unitarity and renormalizability requirements
impose internal symmetries and at least the presence of one scalar
field, under appropriate circumstances~\cite{CornwallLT73,
CornwallLT74}. For heavy vector boson interactions, the
Higgs-mechanism shortcut replaced this wisdom in the textbooks.
Similar bottom-up arguments surface nowadays in \cite[Prob.~9.3 and
Sect.~27.5]{Schwartz14}.}

All that being said, the model expounded here is of course anomalous,
which manifests itself in~$S_3$. The cure is the same as in the
standard treatments. The computation of the chiral anomaly in our
framework will be published elsewhere.

A natural question is: to what extent, on the basis of string
independence of the couplings, chirality of the interaction with
fermions is a generic trait of physics models. We do not have a
comprehensive answer to this. From our treatment here one gathers that
models with only massless bosons like QCD are purely vectorial, on the
one hand. Limits of the SM, like the Georgi--Glashow model and the
Higgs--Kibble model, on the other hand, must exhibit chirality.

\appendix

\section{Proof of Eq.~\eqref{eq:the-secret-garden}}
\label{app:aliquid-obstat}

We prove here the identities 
\begin{align}
\sum_\vf \biggl( d_l \frac{\del S_1}{\del \vf} \vev{\T \vf\chi'} 
+ \frac{\del S_1}{\del \vf} \vev{\T (d_l\vf) \chi'} \biggr)
&= [\T (d_l S_1) \chi']_\tree \,,
\label{eq:secret-garden-1} 
\\
\sum_\psi \biggl( 
\del_\mu \frac{\del Q^\mu}{\del\psi} \vev{\T \psi \chi'} 
+ \frac{\del Q^\mu}{\del\psi} \vev{\T (\del_\mu\psi) \chi'} \biggr)
&= [\T (\del_\mu Q^\mu) \chi']_\tree \,. 
\label{eq:secret-garden-2} 
\end{align}
Using the identity 
$$
d_l S_1 = \sum_\vf \wick:\frac{\del S_1}{\del \vf}\, d_l\vf:, 
$$
the right-hand side of Eq.~\eqref{eq:secret-garden-1} is
$$
\sum_\vf \biggl[ 
\T \wick:\frac{\del S_1}{\del \vf}\, d_l\vf: \chi' \biggl]_\tree
= \sum_\psi \biggl\{ \sum_\vf
\wick:\frac{\del^2 S_1}{\del\vf\del\psi}\, d_l\vf: \biggr\} 
\vev{\T \psi \chi'}
+ \sum_\vf \frac{\del S_1}{\del\vf}\, \vev{\T (d_l\vf) \chi'}. 
$$
But the term in braces is just $d_l\frac{\del S_1}{\del\psi}$. Hence
the right-hand side of the above equation coincides with the left-hand
side of Eq.~\eqref{eq:secret-garden-1}. 

Similarly, using $\del_\mu Q^\mu 
= \sum_\vf \wick:(\del Q^\mu/\del\vf) \,\del_\mu\vf:$,
the right-hand side of Eq.~\eqref{eq:secret-garden-2} becomes
$$
\sum_\psi \biggl\{ \sum_\vf
\wick:\frac{\del^2 Q^\mu}{\del\vf\del\psi}\, \del_\mu\vf: \biggr\} 
\vev{\T \psi \chi'}
+ \sum_\vf \frac{\del Q^\mu}{\del\vf}\, 
\vev{\T (\del_\mu \vf) \chi'}, 
$$
which equals the left-hand side of Eq.~\eqref{eq:secret-garden-2}.

\section{Fermionic crossings}
\label{app:cross-purposes}

The crossings of fermionic type in Section~\ref{sec:at-laeva-malorum} 
are computed as follows. When crossing $\bar e \ga^\mu \nu$ with
$\bar\nu' \ga^k \ga^5 e'$, say, one meets two obstructions of 
type~\eqref{eq:delta-glitch}: contracting the neutrinos gives a factor
$\sO(\ga \nu, \bar\nu') = - \dl(x - x')$, whereas contraction of the 
electrons gives $\sO(e', \bar e \ga) = + \dl(x - x')$. Thus, the 
overall crossing yields a sum of two terms
$$
- \bar e(x) \ga^\ka \ga^5 e(x) \,\dl(x - x')
+ \bar\nu(x) \ga^\ka \ga^5 \nu(x) \,\dl(x - x').
$$
On the other hand, the crossing of $\bar e \ga^\mu \ga^5 e$ with
$\bar e' \ga^5 e'$, say, involving both 
$\sO(\ga e, \bar e')$ and $\sO(e', \bar e \ga)$, gives two equal 
contributions of $\bar e(x)\,e(x)\,\dl(x - x')$ to the total 
obstruction.

There are sixteen kinds of crossings in all, taking account of the
order of the contractions, and the presence or absence of $\ga^\ka$
and/or $\ga^5$ factors. Let $f$ denote a fermion ($\nu$
or~$e$, as the case may be). When computing the crossings, we label 
the contracted terms with stars: either $\ga^\mu f\,\bar f'$ is 
replaced by $\sO(\ga f, \bar f') = - \dl$, or $f'\,\bar f \ga^\mu$ is 
replaced by $\sO(f', \bar f \ga) = + \dl$. In the table which follows,
$\sg$ and $\tau$ denote uncontracted fermions:
\begin{align}
\bar \sg \ga^\mu \2{f} \, \2{\bar f'} \ga^\ka \tau' 
&\leadsto - \bar \sg \ga^\ka \tau \.\dl, 
& \2{\bar f} \ga^\mu \tau \, \bar \sg' \ga^\ka \2{f'} 
&\leadsto + \bar \sg \ga^\ka \tau \.\dl,
\notag \\
\bar \sg \ga^\mu \ga^5 \2{f} \, \2{\bar f'} \ga^\ka \ga^5 \tau' 
&\leadsto - \bar \sg \ga^\ka \tau \.\dl, 
& \2{\bar f} \ga^\mu \ga^5 \tau \, \bar \sg' \ga^\ka \ga^5 \2{f'} 
&\leadsto + \bar \sg \ga^\ka \tau \.\dl,
\notag \\
\bar \sg \ga^\mu \ga^5 \2{f} \, \2{\bar f'} \ga^\ka \tau' 
&\leadsto - \bar \sg \ga^\ka \ga^5 \tau \.\dl, 
& \2{\bar f} \ga^\mu \ga^5 \tau \, \bar \sg' \ga^\ka \2{f'} 
&\leadsto + \bar \sg \ga^\ka \ga^5 \tau \.\dl.
\notag \\
\bar \sg \ga^\mu \2{f} \, \2{\bar f'} \ga^\ka \ga^5 \tau' 
&\leadsto - \bar \sg \ga^\ka \ga^5 \tau \.\dl, 
& \2{\bar f} \ga^\mu \tau \, \bar \sg' \ga^\ka \ga^5 \2{f'} 
&\leadsto + \bar \sg \ga^\ka \ga^5 \tau \.\dl,
\notag \\[2\jot]
\bar \sg \ga^\mu \2{f} \, \2{\bar f'} \tau' 
&\leadsto - \bar \sg \tau \.\dl, 
& \2{\bar f} \ga^\mu \tau \, \bar \sg' \2{f'} 
&\leadsto + \bar \sg \tau \.\dl,
\notag \\
\bar \sg \ga^\mu \ga^5 \2{f} \, \2{\bar f'} \ga^5 \tau' 
&\leadsto + \bar \sg \tau \.\dl, 
& \2{\bar f} \ga^\mu \ga^5 \tau \, \bar \sg' \ga^5 \2{f'} 
&\leadsto + \bar \sg \tau \.\dl,
\notag \\
\bar \sg \ga^\mu \ga^5 \2{f} \, \2{\bar f'} \tau' 
&\leadsto + \bar \sg \ga^5 \tau \.\dl, 
& \2{\bar f} \ga^\mu \ga^5 \tau \, \bar \sg' \2{f'} 
&\leadsto + \bar \sg \ga^5 \tau \.\dl,
\notag \\
\bar \sg \ga^\mu \2{f} \, \2{\bar f'} \ga^5 \tau' 
&\leadsto - \bar \sg \ga^5 \tau \.\dl, 
& \2{\bar f} \ga^\mu \tau \, \bar \sg' \ga^5 \2{f'} 
&\leadsto + \bar \sg \ga^5 \tau \.\dl.
\label{eq:bearing-the-cross} 
\end{align}

\section{Proof of locality of the stringy fields}
\label{app:BW}

We prove here locality in the sense that $A_\mu(x,l)$ and
$A_\al(x',l')$ commute if the strings $\{x + tl\}$ and $\{x' + tl'\}$
are causally disjoint and not parallel. We begin with some geometric
considerations about wedge regions. These are Poincar\'e transforms of
the wedge
$$
W_1 := \set{x \in \bR^4 : x^1 > |x^0|}.
$$
Associated with $W_1$ are the one-parameter group $\La_1(\cdot)$ of
Lorentz boosts which leave $W_1$ invariant, and the reflection $j_1$
across the edge of the wedge. More specifically, $\La_1(t)$ acts as
$$
\twobytwo{\cosh t}{\sinh t}{\sinh t}{\cosh t} 
$$ 
and $j_1$ acts as the reflection on the coordinates $x^0$ and $x^{1}$,
leaving the other coordinates unchanged. For a general wedge
$W = L W_1 = a + \La W_1$ with $L = (a,\La)$, one defines the
corresponding boosts $\La_W(\cdot)$ and reflection $j_W$ by
$$
\La_W(t) := L\, \La_1(t) \,L^{-1}, \quad
j_W := L\, j_1 \,L^{-1}. 
$$
The reflection $j_W$ results from analytic extension of the (entire
analytic) matrix-valued function $\La_W(z)$ at $z = i\pi$.

Note that in the definition of covariance in
Section~\ref{sec:on-a-string} the string direction transforms only
under the homogeneous part of the Poincar\'e transformations. This
leads us to consider the mapping $(a,\La): l \mapsto \La l$ as the
natural action of the Poincar\'e group on the manifold of string
directions. In particular, if $W = a + \La W_1$ then
\begin{equation} 
\La_W(t) l = \La \La_1(t) \La ^{-1} l. 
\label{eq:Lambda-l} 
\end{equation}

\begin{lema} 
\label{lm:string-wedge}
\begin{enumerate}
\item[\textup{(i)}] 
A string $\{x + tl\}$ is contained in the closure of a wedge
$W = a + \La W_1$ if and only if $x$ and~$l$ are contained in the
closures of $W$ and~$\La W_1$ respectively.

\item[\textup{(ii)}] 
Suppose that the strings $\{x + tl\}$ and $\{x' + tl'\}$ are causally
disjoint and not parallel. Then there is a wedge $W$ whose closure
contains $\{x + tl\}$ and whose causal complement contains
$\{x' + tl'\}$.
The corresponding boosts respectively act as 
\begin{equation} 
\La_W(t) l = e^t l \quad\word{and}\quad \La_W(t) l' = e^{-t} l'.
\label{eq:LW-l} 
\end{equation}
\end{enumerate}
\end{lema}

\begin{proof}
Item~(i) is the same as in Lemma~A.1. of~\cite{MundSY06}, whose proof
is valid for any direction $l \in \bR^4$. 

For item~(ii), take $W := \half(x + x') + W_{l,l'}$, where 
$W_{l,l'} := \set{y : (yl) < 0 < (yl')}$. The causal complement of~$W$
is the closure of $\half(x + x') + W_{l',l}$, see~\cite{ThomasW97}.
Furthermore, $l$ is -- up to a factor -- the only lightlike vector
contained in the upper boundary of $W_{l,l'}$ (which is a part of the
lightlike hyperplane $l^\perp$).

Using the elementary fact that $\{x + tl\}$ and $\{x' + tl'\}$ are
causally disjoint if and only if $(x - x')^2 < 0$ and
$\bigl( (x'-x) l \bigr) \geq 0 \geq \bigl( (x' - x) l' \bigr)$, one
readily verifies~\cite{Figueiredo17} that these strings are contained
in the respective wedges $\ovl{W}$ and $W'$, as claimed.

In terms of the lightlike vectors $l_{(\pm)} = (1,\pm 1,0,0)$, the
standard wedge $W_1$ is just $W_{l_{(+)},l_{(-)}}$. Since $l_{(+)}$
is, up to a factor, the only lightlike vector contained in the upper
boundary of $W_1$, the Lorentz transformation $\La$ maps the span of
$l_{(+)}$ onto the span of~$l$. Thus, 
$\La_W(t) l \equiv \La \La_1(t) \La^{-1} l$ is a multiple of
$\La \La_1(t) l_{(+)}$. But one readily verifies that
$\La_1(t) l_{(+)} = e^t\,l_{(+)}$. This proves the first equation in
\eqref{eq:LW-l}. The second is shown analogously, using that $\La$
maps the span of $l_{(-)}$ onto that of~$l'$.
\end{proof}

We now prove locality of the two-point function, recalling first that
the on-shell two-point function for not necessarily coinciding
directions is given, instead of~\eqref{eq:two-pointers-AA}, by 
\begin{equation}
M^{AA}_{\mu\nu}(p,l,l') 
= - g_{\mu\nu} + \frac{p_\mu l_\nu}{(pl)} + \frac{p_\nu l'_\mu}{(pl')}
- \frac{p_\mu p_\nu\,(l\,l')}{(pl) (pl')},
\label{eq:polarized-two-pointer} 
\end{equation}
see~\cite{MundDO17}. Given the two causally disjoint and non-parallel
strings, let $W$ be a wedge whose closure contains $\{x+tl\}$ and
whose causal complement contains $\{x' + tl'\}$ (as in the lemma), and
let $j_W$ and $\La_W(t)$ be the reflection and the boosts,
respectively, corresponding to $W$. Denote by $g_t$ the proper
non-orthochronous Poincar\'e transformation $\La_W(-t)j_W$. By
translation invariance of the two-point function, we may assume that
the edge of $W$ contains the origin. Then $x$ and $l$ are in the
closure of~$W$, while $x'$ and $l'$ lie in the causal complement
of~$W$. This implies that for $t$ in the strip $\bR + i(0,\pi)$ the
imaginary parts of both $g_tx$ and $g_{-t}x'$ lie in the closed
forward light cone -- see, for example, Eq.~(A.7) in~\cite{MundSY06}.

Now consider the relation  
\begin{equation}
\int d\mu(p)\, e^{-i(p(x' - g_t x))}\, M^{AA}_{\al\mu}(p,l',g_t l)
= \int d\mu(p)\, e^{-i(p(x - g_{-t} x'))}\, 
M^{AA}_{\al\mu}(-g_t p,l',g_t l),
\label{eq:Loc-TwoPt} 
\end{equation}
which is verified by applying the transformation $p \mapsto -g_t p$ on
the mass shell. (We use $-g_t$ instead of $g_t$, since the former is
an orthochronous Poincar\'e transformation, while the latter is not
orthochronous and maps the positive onto the negative mass shell.) We
may write $g_t^{-1} = g_{-t}$, since $j_W$ and $\La_W(t)$ commute. We
wish to extend the function $F(t)$ defined by~\eqref{eq:Loc-TwoPt}
analytically into the strip $\bR + i(0,\pi)$. To this end, note that
the Minkowski products of $g_t x$ and $g_{-t}x'$ with a covector $p$
in the mass shell both have positive imaginary parts due to the remark
before Eq.~\eqref{eq:Loc-TwoPt}. This implies that the functions
$|\exp{i(p g_tx)}|$ and $|\exp{i(p g_{-t}x')}|$ are uniformly bounded
by~$1$ over the strip. Furthermore, 
$M^{AA}_{\al\mu}(p,l',g_t l) = M^{AA}_{\al\mu}(p,l',l)$ since
$g_t l = e^t l$ by Eq.~\eqref{eq:LW-l}, and the factor $e^t$ cancels
as can be seen from Eq.~\eqref{eq:polarized-two-pointer}. By the same
token plus covariance, one obtains
$$
M^{AA}_{\al\mu}(-g_t p,l',g_t l) \equiv
(-g_t)_\al{}^\bt M^{AA}_{\bt\nu}(p,-g_{-t}l',-l) (-g_{-t})_\mu{}^\nu
= (g_t)_\al{}^\bt M^{AA}_{\bt\nu}(p,l',l) (g_{-t})_\mu{}^\nu.
$$ 

These facts imply that $F(t)$ has an analytic extension into the
strip, and Eq.~\eqref{eq:Loc-TwoPt} holds, by the Schwarz reflection
principle, also at $t = i\pi$. But $g_{\pm i\pi} = 1$, and thus at
$t = i\pi$ the left-hand side of Eq.~\eqref{eq:Loc-TwoPt} reduces, up
to a factor $(2\pi)^3$, to the vacuum expectation value
$\vev{A_\al(x',l') A_\mu(x,l)}$. On the right-hand side, one verifies
that $M^{AA}_{\al\mu}(p,l',l) = M^{AA}_{\mu\al}(p,l,l')$. Thus, at
$t = i\pi$ the right side of~\eqref{eq:Loc-TwoPt} reduces, up to a
factor $(2\pi)^3$, to $\vev{A_\mu(x,l) A_\al(x',l')}$. In short,
Eq.~\eqref{eq:Loc-TwoPt} at $i\pi$ is just the locality of the
two-point functions. This implies locality of the fields by a standard
argument in the proof of the Jost--Schroer theorem~\cite{StreaterW64}.

\section{A model of leptons}
\label{app:y-ole}

Engineering the GWS model from our formalism is not overly desirable.
But we do it here, as promised in the introduction. Let us reconsider
the three first lines of expression \eqref{eq:ave-atque-vale}. We
begin by introducing the notation
$$
\Psi_L := \begin{pmatrix} \nu_L \\ e_L \end{pmatrix} 
:= \begin{pmatrix} \half(1 - \ga^5) \nu \\[\jot]
\half(1 - \ga^5) e \end{pmatrix}.
$$
First,
$$
- \frac{1}{\sqrt2} W_{-\mu} \bar e \ga^\mu \frac{(1 - \ga^5)}{2} \nu
= - \frac{1}{\sqrt2}\, 
\Psibar_L \ga^\mu \twobytwo{0}{0}{W_{-\mu}}{0} \Psi_L
= - \half \Psibar_L \ga^\mu W_{-\mu} \tau_- \Psi_L;
$$
where $\tau_\pm = (\tau_1 \pm i\tau_2)/\sqrt2$, with $\tau_i$ denoting
here the Pauli matrices. Similarly,
$$
- \frac{1}{\sqrt2} W_{+\mu} \bar\nu \ga^\mu \frac{(1 - \ga^5)}{2} e
= - \half \Psibar_L \ga^\mu W_{+\mu} \tau_+ \Psi_L.
$$
The first two terms in~\eqref{eq:ave-atque-vale} are therefore of the
form
\begin{equation}
- \half g \Psibar_L \ga^\mu (W_{+\mu} \tau_+ + W_{-\mu} \tau_-) \Psi_L
= -\half g \Psibar_L \ga^\mu (W_{1\mu}\tau_1 + W_{2\mu}\tau_2) \Psi_L.
\label{eq:SM-technological} 
\end{equation}
Knowing, as we know, that the interaction is governed by a $U(2)$
symmetry, it is tempting to regard $\nu$ and~$e$ as isospin
components valued $+\half$ and~$-\half$, respectively. The
``right-handed leptons'' $e_R := \half(1 + \ga^5) e$ and
$\nu_R := \half(1 + \ga^5) \nu$ are isospin singlets. 

Denote by $Q$ the electric charge, so that $Q(e) = -1$ and 
$Q(\nu) = 0$, and isospin by~$I_3$. Observe that, putting
$\Psi = \Psi_L + \Psi_R$, the next four terms
of~\eqref{eq:ave-atque-vale} are rendered into:
\begin{align}
-g \sin\Th\, \Psibar \ga^\mu (A_\mu - Z_\mu \tan\Th) Q \Psi
- \frac{g}{\cos\Th}\, \Psibar_L \ga^\mu Z_\mu I_3 \Psi_L \,.
\label{eq:SM-ideological} 
\end{align}

In order to translate this into the received framework, with its
``covariant gauge transformation'' technology, we now introduce the
unobservable fields
$$
\begin{aligned}
W_{3\mu} &:= \phantom{+} \cos\Th\, Z_\mu + \sin\Th\, A_\mu
\\
B_\mu &:= - \sin\Th\, Z_\mu + \cos\Th\, A_\mu
\end{aligned}
\word{with inversion}
\begin{aligned}
A_\mu &= \phantom{+} \cos\Th\, B_\mu + \sin\Th\, W_{3\mu} 
\\
Z_\mu &= - \sin\Th\, B_\mu + \cos\Th\, W_{3\mu} \,.
\end{aligned}
$$
Then, with $g_B := g \tan\Th$, we can rewrite
\eqref{eq:SM-ideological} as
\begin{align}
- g_B \Psibar \ga^\mu B_\mu Q \Psi 
+ g_B \Psibar_L \ga^\mu B_\mu I_3 \Psi_L
- \half g \Psibar_L \ga^\mu W_{3\mu} \tau_3 \Psi_L \,.
\label{eq:ultra-ideological} 
\end{align}
One can now bring in the convention
$$
Y = 2(Q - I_3), \word{that is:}
Y(e_L) = Y(\nu_L) = -1; \ Y(e_R) = -2, \ Y(\nu_R) = 0.
$$
Then the first two summands in~\eqref{eq:ultra-ideological} are
rewritten as $-\half g_B \Psibar \ga^\mu B_\mu Y \Psi$; while the last
one together with the right hand side of~\eqref{eq:SM-technological}
yields $-\half g \Psibar_L (\ga^\mu \bW_\mu \cdot \tau) \Psi_L$.

\textit{In fine}, we have manufactured the interaction parts of the
GWS Lagrangian.

\section*{Acknowledgments}

We thank Michael D\"utsch first and foremost for discussions; perhaps
without him this paper would never have been written. The reviewers'
comments were helpful and improved the content and presentation of the
paper. We are also grateful to Alejandro Ibarra and Bert Schroer for
lively exchanges of views. 

This research was generously helped by the program ``Research in
Pairs'' of the Mathe\-matisches Forschungsinstitut Oberwolfach in
November~2015. The project has received funding from the European
Union's Horizon~2020 research and innovation programme under the Marie
Sk{\l}odowska-Curie grant agreement No.~690575. JM was partially
supported by CNPq, CAPES, FAPEMIG and Finep, and thanks the
Universidad de Costa~Rica for hospitality. JMG-B received funding from
Project FPA2015--65745--P of MINECO/Feder. JCV acknow\-ledges support
from the Vicerrector\'ia de Investigaci\'on of the Universidad de
Costa~Rica.


\begin{thebibliography}{42}

\footnotesize

\bibitem{Phaidros}
Socrates:
``Would a sensible husbandman, who has seeds which he cares for and
which he wishes to bear fruit, plant them with serious purpose in the
heat of summer in some garden of Adonis\tightldots?''

\bibitem{Quigg09}
C. Quigg,
``Unanswered questions in the electroweak theory'',
Ann. Rev. Nucl. Part. Sci. \textbf{59} (2009), 505--555.

\bibitem{LynnS15}
B. W. Lynn and G. D. Starkman,
``Global $SU(3)_C \x SU(2)_L \x U(1)_Y$ linear sigma model with
Standard Model fermions: axial-vector Ward Takahashi identities, the
absence of Brout--Englert--Higgs mass fine tuning, and the decoupling
of certain heavy particles, due to the Goldstone theorem'',
arXiv:1509.06199.

\bibitem{PeskinS95}
M. E. Peskin and D. V. Schroeder,
\textit{An Introduction to Quantum Field Theory},
Addison-Wesley, Reading, MA, 1995.

\bibitem{Marshak93}
R. E. Marshak,
\textit{Conceptual Foundations of Modern Particle Physics},
World Scientific, Singapore, 1993.

\bibitem{MundSY04}
J. Mund, B. Schroer and J. Yngvason,
``String-localized quantum fields from Wigner representations'',
Phys. Lett. B\,\textbf{596} (2004), 156--162.

\bibitem{MundSY06}
J. Mund, B. Schroer and J. Yngvason,
``String-localized quantum fields and modular localization'',
Commun. Math. Phys. 268 (2006), 621--672.

\bibitem{Mandelstam62}
S. Mandelstam,
``Quantum electrodynamics without potentials'',
Ann. Phys. (NY) \textbf{19} (1962), 1--24.

\bibitem{Steinmann84}
O. Steinmann,
``Perturbative QED in terms of gauge invariant fields'',
Ann. Phys. (NY) \textbf{157} (1984), 232--254.

\bibitem{Schroer16}
B. Schroer,
``Beyond gauge theory: positivity and causal localization in the 
presence of vector mesons'',
Eur. Phys. J. C\,\textbf{76} (2016) 378.

\bibitem{Weinberg95I}
S. Weinberg,
\textit{The Quantum Theory of Fields I},
Cambridge University Press, Cambridge, 1995.

\bibitem{NikolovST14}
N. M. Nikolov, R. Stora and I. Todorov,
``Renormalization of massless Feynman amplitudes in configuration
space'',
Rev. Math. Phys. 26 (2014), 1430002 (65 pages).

\bibitem{Leto}
J. C. V\'arilly and J. M. Gracia-Bond\'ia,
``Stora's fine notion of divergent amplitudes'',
Nucl. Phys. B \textbf{912} (2016), 28--37.

\bibitem{Wigner39}
E. P. Wigner,
``On unitary representations of the inhomogeneous Lorentz group'',
Ann. Math. \textbf{40} (1939), 149--204.

\bibitem{Yngvason70}
J. Yngvason,
``Zero-mass infinite spin representations of the Poincar\'e group and 
quantum field theory'',
Commun. Math. Phys. \textbf{18} (1970), 195Ñ-203.

\bibitem{IversonM71}
G. J. Iverson and G. Mack,
``Quantum fields and interactions of massless particles: the 
continuous spin case'',
Ann. Phys. (NY) \textbf{64} (1971), 211--253.

\bibitem{Rehren17}
K.-H. Rehren,
``Pauli--Luba\'nski limit and stress-energy tensor for infinite-spin
fields'',
JHEP \textbf{1711} (2017), 130.

\bibitem{Peskin16}
M. Peskin, 
``Standard Model and symmetry breaking'',
talk given at the Latin American conference on High Energy Physics:
``Particles and Strings~II'', Havana, July 2016.

\bibitem{DuetschS99}
M. D\"utsch and G. Scharf,
``Perturbative gauge invariance: the electroweak theory''
Ann. Phys. (Leipzig) \textbf{8} (1999), 359--387.

\bibitem{AsteSD99}
A. Aste, G. Scharf and M. D\"utsch,
``Perturbative gauge invariance: electroweak theory II''
Ann. Phys. (Leipzig) \textbf{8} (1999), 389--404.

\bibitem{Scharf16}
G. Scharf,
\textit{Gauge Field Theories: Spin One and Spin Two},
Dover, Mineola, NY, 2016.

\bibitem{Stora05}
R. Stora,
``From Koszul complexes to gauge fixing'',
in \textit{50 Years of Yang--Mills Theory}, G.~'t\,Hooft, ed.,
World Scientific, Singapore, 2005; pp.~137--167.

\bibitem{Mund09}
J. Mund,
``String-localized quantum fields, modular localization, and gauge
theories'',
in \textit{New Trends in Mathematical Physics},
V. Sidoravi\v{c}ius, ed.,
Springer, Heidelberg, 2009; pp.~495--508.

\bibitem{Duch17}
P. Duch,
``Massless fields and adiabatic limit in quantum field theory'',
Ph.~D. thesis, Jagiellonian University, Cracow, summer 2017;
arXiv:1709.09907.

\bibitem{Leibbrandt87}
G. Leibbrandt,
``Introduction to noncovariant gauges'',
Rev. Mod. Phys. \textbf{59} (1987), 1067--1119.

\bibitem{Figueiredo17} 
F. Figueiredo,
``Lightlike string-localized free quantum fields for massive bosons'',
M.~Sc. thesis, Universidade Federal de Juiz de Fora, 2017. 

\bibitem{PlaschkeY12}
M. Plaschke and J. Yngvason,
``Massless, string localized quantum fields for any helicity'',
J. Math. Phys. \textbf{53} (2012), 042301.

\bibitem{MundDO17}
J. Mund and E. T. de~Oliveira,
``String-localized free vector and tensor potentials for massive
particles with any spin: I.~Bosons'',
Commun. Math. Phys. \textbf{355} (2017), 1243--1282.

\bibitem{MundRS17b}
J. Mund, K.-H. Rehren and B. Schroer,
``Helicity decoupling in the massless limit of massive tensor fields'',
Nucl. Phys. B \textbf{924} (2017), 699--727.

\bibitem{BogoliubovS80}
N. N. Bogoliubov and D. V. Shirkov,
\textit{Introduction to the Theory of Quantized Fields}, 3rd~edition,
Wiley, New York, 1980.

\bibitem{EpsteinG73}
H. Epstein and V. Glaser,
``The role of locality in perturbation theory'',
Ann. Inst. Henri Poincar\'e A\,\textbf{19} (1973), 211--295.

\bibitem{Schwartz66}
L. Schwartz,
\textit{Th\'eorie des distributions},
Paris, Hermann, 1966.

\bibitem{Zichichi84}
A. L. Zichichi \textit{et~al},
``Special section on symmetries and gauge invariance'',
in \textit{Gauge Interactions}, A.~L. Zichichi, ed.,
Plenum Press, New York, 1984; pp.~725--740.

\bibitem{Okun91}
Lev B. Okun,
``From pions to wions'',
in \textit{The Relations of Particles},
World Scientific, Singapore, 1991; pp.~31--45.

\bibitem{Scheck12}
F. Scheck,
\textit{Electroweak and Strong Interactions: Phenomenology, Concepts, 
Models},
Springer, Berlin, 2012.

\bibitem{Nagashima13}
Y. Nagashima,
\textit{Elementary Particle Physics 2: Foundations of the Standard 
Model},
Wiley, Singapore, 2013.

\bibitem{WienNotes}
J. Mund,
``String-localized massive vector bosons in interaction'',
in preparation.

\bibitem{MundS17}
J. Mund and B. Schroer,
``How the Higgs potential got its shape'',
forthcoming.

\bibitem{CornwallLT73}
J. M. Cornwall, D. N. Levin and G. Tiktopoulos,
``Uniqueness of spontaneously broken gauge theories'',
Phys. Rev. Lett. \textbf{30} (1973), 1268--1270.

\bibitem{CornwallLT74}
J. M. Cornwall, D. N. Levin and G. Tiktopoulos,
``Derivation of gauge invariance from high-energy unitarity bounds on
the $\bS$-matrix'',
Phys. Rev. D \textbf{10} (1974), 1145--1167.

\bibitem{Schwartz14}
M. D. Schwartz,
\textit{Quantum Field Theory and the Standard Model},
Cambridge University Press, Cambridge, 2014.

\bibitem{ThomasW97}
L. J. Thomas and E. H. Wichmann,
``On the causal structure of Minkowski spacetime'',
J. Math. Phys. \textbf{38} (1997), 5044--5086.

\bibitem{StreaterW64}
R. F. Streater and A. S. Wightman,
\textit{PCT, Spin and Statistics, and All That},
W. A. Benjamin, New York, 1964.

\end{thebibliography}
\end{document}